\documentclass[11pt]{article}

\usepackage{hyperref}

\usepackage{latexsym,xcolor,graphicx,mathtools,amssymb,amsthm,comment}

\usepackage{microtype}

\usepackage{fullpage}

\newtheorem{theorem}{Theorem}[section]
\newtheorem{lemma}[theorem]{Lemma}
\newtheorem{proposition}[theorem]{Proposition}
\newtheorem{corollary}[theorem]{Corollary}
\newtheorem{assumption}{Assumptions}

\theoremstyle{definition}
\newtheorem*{definition}{Definition}

\def\cplx{\mathbb{C}}

\def\reals{\mathbb{R}}
\let\eps\varepsilon
\let\bd\partial
\def\A{\mathcal{A}}

\DeclarePairedDelimiter{\abs}{\lvert}{\rvert}

\def\etal{\textit{et~al.}}

\DeclareMathOperator{\poly}{poly}
\DeclareMathOperator{\res}{res}

\newcommand{\ignore}[1]{}

\begin{document}

\title{Testing Polynomials for Vanishing on \\
Cartesian Products of Planar Point Sets: \\
Collinearity Testing and Related Problems\thanks{%
Work by Boris Aronov was partially supported by NSF grant CCF-15-40656 
and by grant 2014/170 from the US-Israel Binational Science Foundation.
Work by Esther Ezra was partially supported by NSF CAREER under grant 
CCF:AF-1553354 and by Grant 824/17 from the Israel Science Foundation.
Work by Micha Sharir was partially supported by ISF Grant 260/18, 
by grant 1367/2016 from the German-Israeli Science Foundation (GIF), and by
Blavatnik Research Fund in Computer Science at Tel Aviv University.}}


\author{Boris Aronov\thanks{%
Department of Computer Science and Engineering,
Tandon School of Engineering,
New York University,
Brooklyn, NY 11201, USA,
\textsf{boris.aronov@nyu.edu},
{http://orcid.org/0000-0003-3110-4702}}%
\and
Esther Ezra\thanks{%
School of Computer Science, Bar Ilan University, Ramat Gan, Israel,
\textsf{ezraest@cs.biu.ac.il},
{https://orcid.org/0000-0001-8133-1335}}%
\and
Micha Sharir\thanks{%
School of Computer Science, Tel Aviv University, Tel~Aviv, Israel,
\textsf{michas@tau.ac.il},
{http://orcid.org/0000-0002-2541-3763}}
}


\maketitle

\begin{abstract}
We present subquadratic algorithms, in the algebraic
decision-tree model of computation, for detecting whether there exists a triple of points,
belonging to three respective sets $A$, $B$, and $C$ of points 
in the plane, that satisfy a certain polynomial equation or two equations. 
The best known instance of such a problem is testing
for the existence of a collinear triple of points in
$A\times B\times C$, a classical 3SUM-hard problem that
has so far defied any attempt to obtain a subquadratic solution, whether in the
(uniform) real RAM model, or in the algebraic decision-tree model. 
While we are still unable to solve this problem,
in full generality, in subquadratic time, we obtain such a solution,
in the algebraic decision-tree model, that uses only roughly
$O(n^{28/15})$ constant-degree polynomial sign tests,
for the special case where two of the sets lie on two respective one-dimensional curves
and the third is placed arbitrarily in the plane.
Our technique is fairly general, and applies to many other problems
where we seek a triple that satisfies a single polynomial equation,
e.g., determining whether $A\times B\times C$ contains a triple 
spanning a unit-area triangle.
This result extends recent work by Barba \etal~(2017)
and by Chan (2018), where all three sets $A$,~$B$, and~$C$ 
are assumed to be one-dimensional. 

As a second application of our technique, we again have three $n$-point sets
$A$, $B$, and $C$ in the plane, and we want to
determine whether there exists a triple $(a,b,c) \in A\times B\times C$
that simultaneously satisfies two independent real polynomial equations. For example,
this is the setup when testing for collinearity in the complex plane,
when each of the sets $A$, $B$, $C$ lies on some constant-degree algebraic curve.
We show that problems of this kind can be solved with roughly 
$O(n^{24/13})$ constant-degree polynomial sign tests. 
\end{abstract}

\section{Introduction}
\label{sec:intro}

\paragraph*{General background.}
Let $A$, $B$, and $C$ be three $n$-point sets in the plane.
We want to determine whether there exists a triple of points 
$(a,b,c)\in A\times B\times C$ that satisfy one or two prescribed polynomial equations.
An example of such a scenario, with a single vanishing polynomial, is to determine whether $A\times B\times C$
contains a collinear triple of points. This classical problem is at least as hard as the 3SUM problem~\cite{GO}, 
in which we are given three sets $A$, $B$, and $C$, each consisting of 
$n$ real numbers, and we want to determine whether there exists a triple 
of numbers $(a,b,c)\in A\times B\times C$ that add up to zero. 

The 3SUM problem itself, conjectured for a long time to require $\Omega(n^2)$ 
time, was shown by Gr{\o}nlund and Pettie~\cite{GP} (with further improvements by Chan~\cite{Chan-18})
to be solvable in very slightly subquadratic time. 
Moreover, in the \emph{linear decision-tree model}, in which we only count
linear sign tests involving the input point coordinates (and do not allow any other
operation to access the input explicitly), Gr{\o}nlund and Pettie
have improved the running time to nearly $O(n^{3/2})$ 
(see also~\cite{Fre,GS} for subsequent slight speedups), 
which has been drastically further improved (still in the linear decision-tree model)
to $O(n\log^2n)$ time by Kane \etal~\cite{KLM}.

In contrast, no subquadratic algorithm is known for the collinearity 
detection problem, either in the standard real RAM model (also
known as the \emph{uniform} model) or in the decision-tree model; see~\cite{BCILOS-17} for a discussion.
In the uniform model, the problem can be solved in~$O(n^2)$~time.
The primitive operation needed to test for collinearity of a specific triple is 
the so-called \emph{orientation test},
in which we test for the sign of a quadratic polynomial in the six 
coordinates of a triple of points in $A\times B\times C$ (see Eq.~\eqref{eq:matrix} 
below for the concrete expression).
Consequently, it is natural (and apparently necessary) to use
the more general \emph{algebraic decision-tree model}, 
in which each comparison is a sign test of some constant-degree polynomial in the
coordinates of a constant number of input points;
see \cite{BenOr83,PS-85} and below. 

\paragraph*{The problems, in more detail.} 

In this paper we consider two main variants of testing a polynomial,
or polynomials, for vanishing on a triple Cartesian product.
The main motivation for the present study is the aforementioned collinearity testing 
question. We present the problem in a wider context, where we are given three 
sets~$A$, $B$, and~$C$, each consisting of $n$ points in the plane, and we consider two scenarios:

\medskip
\noindent
\textbf{(a) A single vanishing polynomial.}
Given a single constant-degree irreducible $6$-variate real polynomial $F$,
determine whether there exists a triple~$(a,b,c)\in A\times B\times C$
such that $F(a,b,c) = 0$.

\medskip
\noindent
\textbf{(b) A pair of vanishing polynomials.}
Given a pair $F$, $G$ of constant-degree irreducible and independent 
$6$-variate real polynomials, determine whether there exists a triple 
$(a,b,c)\in A\times B\times C$ such that $F(a,b,c) = G(a,b,c) = 0$.

\medskip
Alternatively, in both types of problems, we may want to count the number of 
triples $(a,b,c)\in A\times B\times C$ that satisfy this equation or 
equations, or report all such triples. We only 
consider the existence problem, but the techniques can easily 
be adapted to handle the other variants, with comparable
bounds, in the algebraic decision-tree model.

As can be expected, our results are stronger for the vanishing pair problem in (b).
That is, requiring the triple $(a,b,c)$ to satisfy two equalities 
facilitates a more efficient solution. In contrast, the 
collinearity testing problem in the real plane, as well as 
more general instances of a single vanishing polynomial in (a), 
seem harder to solve efficiently. As we spell out below,
we can solve problems of type (a) in subquadratic time, in the algebraic
decision-tree model, only for suitably restricted input sets.

A special (but natural) case of the problem with two polynomial constraints 
is where each of the sets $A$, $B$, $C$ consists of $n$ complex numbers, and we 
want to test, given a \emph{single} constant-degree trivariate polynomial 
$H \colon \mathbb{C}^3 \to \mathbb{C}$, whether
there is a triple $(a,b,c) \in A\times B\times C$ such that\footnote{%
    Over the reals, $H$ induces two polynomial equations, 
    one for the real part and the other for the imaginary part,
    so this is indeed a special case of the polynomial vanishing pair problem.}
$H(a,b,c) = 0$. This is an extension of the problem studied by 
Barba et al.~\cite{BCILOS-17} over the reals. 
A simple instance of this setup is collinearity testing
in the complex plane, where each of the sets $A$, $B$, $C$ lies on some
suitably parameterized constant-degree algebraic curve; see below for more specific details.

In an earlier version of this work~\cite{AES:prev}, we also mentioned two other instances 
of the polynomial pair vanishing problem, which are two variants of the problem of testing 
for the existence of similar triangles that are determined by $A$, $B$, and $C$.
However, these instances involve a complex polynomial that is either linear (in one instance) 
or can be made linear (in the other instance) after a certain transformation of the input (see~\cite{AES:prev}). 
Following the very recent analysis of Aronov and Cardinal~\cite{AC-20}, 
the case where $H$ is linear, or the more general case where each of our 
real polynomials $F$, $G$ is linear, can be solved, in the \emph{linear} 
decision-tree model, in nearly linear time, using the technique of Kane et al.~\cite{KLM}.
We therefore do not handle these instances in the present version.

\paragraph*{Comments on the purely one-dimensional setup.}

Questions of the type studied in this paper are extensions to higher dimensions
of the algorithmic counterparts of the classical
problems in combinatorial geometry, studied by~Elekes and R\'onyai~\cite{ER}
and by~Elekes and Szab\'o~\cite{ESz}, and later improved, respectively,
by~Raz, Sharir, and Solymosi~\cite{RSS}, and by~Raz, Sharir, and De~Zeeuw~\cite{RSdZ}.
In these problems $A$, $B$, and $C$ are sets of real (or complex) numbers,
and the goal is to bound the number of zeroes that a trivariate real 
(or complex) constant-degree polynomial~$F$ can have on $A\times B\times C$.
As these results show, if $F$ does not have a special, group-like, form,
the number of zeroes is only~$O(n^{11/6})$ (otherwise it can be $\Theta(n^2)$).
This raises the question whether, for polynomials~$F$ that do not have the special form, 
the problem (for one-dimensional sets $A$, $B$, $C$) should be solvable in strictly subquadratic time 
in the uniform model. Strictly subquadratic solutions have recently been obtained in
Barba \etal~\cite{BCILOS-17} in the algebraic decision-tree model,
regardless of whether $F$ does or does not have the special form, for \emph{one-dimensional} 
sets $A$, $B$, $C$, where the running time of the algorithm is close to $O(n^{12/7})$. 
The same approach, combined with more involved algorithmic techniques, 
yields (in \cite{BCILOS-17}) an algorithm in the uniform model that runs in 
$O(n^2(\log\log n)^{3/2}/\log^{1/2}n)$ time. The latter running time has been 
slightly improved to~$O(n^2(\log\log n)^{O(1)}/\log^2n)$ by Chan~\cite{Chan-18}.

Given this apparent (polynomial) hardness of computation, already for one-dimensional point sets in the uniform model,
our goal is thus to obtain a significantly subquadratic solution
(that is, a solution with running time $O(n^{2-c})$, for some constant $c>0$)
in the \emph{algebraic decision-tree model}.
Here we only pay for sign tests that involve the input point coordinates, where
each such test determines the sign of some real polynomial of constant degree
in a constant number of variables. All other operations cost nothing in this model,
and are assumed not to access the input explicitly. For example, detecting
whether some triple of points $(a_1,a_2)\in A$, $(b_1,b_2)\in B$, $(c_1,c_2)\in C$ 
is collinear is done by examining whether the determinant
\begin{equation}
  \label{eq:matrix}
\begin{vmatrix}
1 & a_1 & a_2 \\
1 & b_1 & b_2 \\
1 & c_1 & c_2
\end{vmatrix} 
\end{equation}
is zero. This determinant (whose sign in general determines the orientation of the triple $(a,b,c)$)
is a quadratic polynomial in the six coordinates $(a_1,a_2,b_1,b_2,c_1,c_2)$. We note though
that using only these orientation tests does not yield a worst-case subquadratic solution 
even in the algebraic decision-tree model (see the lower bound argument in~\cite{ES} and see below),
and therefore one needs to resort to somewhat more general polynomial sign tests in order to achieve 
subquadratic solutions in the algebraic decision-tree model.

\paragraph*{Concrete problems in the two-dimensional setup.}
  
Each of the two general questions studied here (of one or two vanishing polynomials)
arises in various concrete problems in computational geometry. For the case of a
single vanishing polynomial, collinearity testing is a fairly famous (or should we say, notorious) example.
Other problems include testing for the existence
of a triangle~$\Delta abc$, for $(a,b,c)\in A\times B\times C$,
that has a given area, or perimeter, or 
circumscribing disk of a given radius, and so on.

Collinearity testing also serves as an instance of the vanishing polynomial pair problem.
The sets $A$, $B$, $C$ are now sets of points in the complex plane $\cplx^2$, each consisting
of $n$ points and lying on some constant-degree algebraic curve of its own,
and the goal is to determine whether $A\times B\times C$ contains a collinear triple.
We will assume in this paper that the curves $\gamma_A$, $\gamma_B$, and $\gamma_C$
that contain, respectively, $A$, $B$, and $C$ are represented parametrically 
by equations of the form $(w,z) = (f_A(t),g_A(t))$, $(w,z) = (f_B(t),g_B(t))$,
and~$(w,z) = (f_C(t),g_C(t))$, where $t$ is a complex parameter and
$f_A$, $g_A$, $f_B$, $g_B$, $f_C$, and $g_C$
are constant-degree univariate (complex) polynomials. (A special case is when
the curves are given by explicit equations of the form $w = f_A(z)$, $w = f_B(z)$, 
and $w = f_C(z)$.) The more general setup in which
$f_A$, $g_A$, $f_B$, $g_B$, $f_C$, and $g_C$ are constant-degree rational functions
can also be handled with a few easy and straightforward modifications.

Given this parameterization, the points of each of the sets
can be represented as points in the real plane (representing the complex numbers $t$),
and the complex polynomial whose vanishing asserts collinearity of a triple 
$a = (z_a,w_a)$, $b = (z_b,w_b)$, $c = (z_c,w_c)$, is
\begin{equation}
  \label{eq:colcomp}
H(t_a,t_b,t_c) \coloneqq
\begin{vmatrix}
1 & f_A(t_a) & g_A(t_a) \\
1 & f_B(t_b) & g_B(t_b) \\
1 & f_C(t_c) & g_C(t_c)
\end{vmatrix} ,
\end{equation}
where $t_a$, $t_b$, $t_c$ are the parameters that specify $a$, $b$, $c$, respectively,
so its real and imaginary components form the pair of real polynomials that have to vanish.
The more general case, where the curves are given by implicit equations of the form
$F_A(z,w) = 0$, $F_B(z,w) = 0$, $F_C(z,w) = 0$, or where the functions $f_A,\ldots,g_C$
are general algebraic functions, will not be addressed in this paper,
although we believe that it can also be handled, using substantially more involved algebraic techniques.

\paragraph*{Our results.}

After setting up the technical machinery that our analysis requires, 
in Sections~\ref{sec:prelim} and~\ref{sec:hierarchical} (and also in 
Appendix~\ref{app:hierarchical}), we first consider, 
in Section~\ref{sec:col}, the problem of ``$2{\times}1{\times}1$-dimensional'' (real) collinearity testing,
meaning that $A$ is an arbitrary set of points in the real plane, but each of $B$ and $C$
lies on some respective constant-degree algebraic curve $\gamma_B$, $\gamma_C$, parameterized as above.
We show that this restricted version of the problem can be solved in the
algebraic decision-tree model with $O(n^{28/15+\eps})$ polynomial
sign tests, for any $\eps > 0$ (where again the constant of proportionality
depends on $\eps$). For this we also need to ensure that the underlying polynomial 
has the so-called ``good fibers'' property---see Section~\ref{sec:prelim} for the 
definition and Section~\ref{sec:col} for further details.
We present a general procedure for ensuring this property for the single vanishing
polynomial problem in Section~\ref{sec:col}. The technique extends naturally to 
similar problems involving a single vanishing polynomial, such as determining 
whether $A\times B\times C$ spans a unit-area triangle, or a triangle with other
similar properties, as mentioned above.

We still do not have a subquadratic solution, even in the algebraic 
decision-tree model, to the unconstrained (referred to as $2{\times}2{\times}2$-dimensional)
collinearity testing problem, or even for the more restricted $2{\times}2{\times}1$-dimensional
scenario, where only one of the sets is constrained to lie on a given curve. 
The techniques that we use for the $2{\times}1{\times}1$ version can be 
extended to the general unconstrained (or less constrained) case, but 
they become too inefficient, and actually result in \emph{superquadratic} 
algorithms; see Sections~\ref{sec:col} and~\ref{sec:fg} for more details.
As shown by Erickson and Seidel~\cite{ES}, if the
\emph{only} sign tests that we allow in the decision tree are orientation tests,
then $\Omega(n^2)$ tests are needed in the worst case. The solution
presented here uses other sign tests, making it more powerful (and more efficient).

We then consider in Section~\ref{sec:fg} the problem of testing for a vanishing pair of polynomials, 
which includes the complex collinearity testing problem.
We show that such problems can be solved, in the algebraic decision-tree model, 
with $O(n^{24/13+\eps})$ polynomial sign tests, for any $\eps>0$ (with the 
constant of proportionality depending on $\eps$), where each test involves 
a real polynomial of constant degree in a constant number of variables, 
which in general might be more involved than the given pair of polynomials 
(like those arising in complex collinearity testing). For 
the analysis, we need to assume that the pair of polynomials $F$,~$G$ 
have ``good fibers'' and some additional properties (which they do in 
the collinearity testing problem). We remark that, unlike Section~\ref{sec:col},
we do not have a general procedure that ensures this property,
and ad-hoc techniques are needed
for verifying that this property holds for each concrete instance of the problem;
see Section~\ref{sec:fg}. 

We also consider an extension of the setup of Section~\ref{sec:col} to higher dimensions.
Specifically, we study collinearity testing in real $d$-dimensional spaces, where we assume that 
each of~$B$ and~$C$ lies on a hyperplane.
Our solution is based on projections of the input onto lower-dimensional subspaces,
and achieves the same asymptotic performance as in the plane.
This result is presented in Section~\ref{sec:higher_dim}.

In the earlier version~\cite{AES:prev} of this work, we have also sketched two more general extensions 
of our technique to the vanishing single-polynomial and polynomial-pair problems 
in $d$ dimensions. In the first extension, we assume that $B$ and $C$ each lies 
on a $(d-1)$-dimensional surface, while $A$ is an arbitrary set of points in 
$\reals^d$, and we seek a triple $(a,b,c)$ in $A\times B\times C$ that satisfies 
$d-1$ independent polynomial equations. Collinearity testing in (real) $d$-dimensional space, 
for input sets restricted as above, is a special instance of this setup. 
In the second extension, each of $A$, $B$, $C$ is an arbitrary set of 
points in $\reals^d$, and we seek a triple in $A\times B\times C$ that 
satisfies $d$ independent polynomial equations. The algorithms sketched in \cite{AES:prev}
were incomplete, and completing them into detailed rigorous solutions 
will make the paper considerably longer and more technical, so we omit them in this version.
We only comment that the sketches in \cite{AES:prev} indicate that the first extension can 
be solved in the algebraic decision-tree model in time $O\left(n^{2-2/(12d^2-20d+7)+\eps}\right)$,
and the second one in time $O(n^{2-2/(6d+1)+\eps})$, for any $\eps>0$.
Both bounds match the corresponding results (and contexts) in the plane, when $d=2$.

\smallskip
\noindent\textbf{Remark.}
While there are common features in 
the high-level approaches in this work and in Barba et al.~\cite{BCILOS-17}, 
the actual analysis in this paper becomes more involved and requires new methods and techniques.
We are not aware of a simple extension of 
the analysis of Barba~\etal~\cite{BCILOS-17} or of Chan~\cite{Chan-18} 
to the problems studied in this paper. A main technique in our arsenal,
which we believe to be essential, is to consider the Cartesian product 
of polynomial partitionings, and algorithmic implementations thereof, 
which are based on so-called hierarchical polynomial partitions.
This technique, and several other novel ingredients, will be elaborated below. 
 
The $2{\times}1{\times}1$ case of problems involving a single vanishing polynomial,
considered in Section~\ref{sec:col}, has an alternative subquadratic, albeit less
efficient, solution, using simpler considerations, which somewhat resemble 
the analysis in~\cite{BCILOS-17}.
We present this alternative technique in Section~\ref{sec:alternative}. 

We also comment that Chan~\cite{Chan-18} addresses several related geometric 3SUM-hard problems,
among which is a variant of dual collinearity testing: Given three sets $A$, $B$, and $C$ 
of line segments in the plane, where the segments in $A$ are pairwise disjoint, and so are
the segments in $B$ and in $C$, decide whether there exists a triple of segments in 
$A\times B \times C$ that meet at a common point.
Although Chan's technique results in a slightly subquadratic algorithm in the RAM model, 
and is also claimed (without details) to yield a truly subquadratic algorithm\footnote{%
  By this we mean an algorithm whose running time is $O(n^{2-c})$ for some constant $c > 0$.} \
in the algebraic decision-tree model,
the disjointedness assumptions significantly restrict the problem, so, 
to quote~\cite{Chan-18}, ``it remains open whether there is a subquadratic algorithm 
for the degeneracy testing for $n$ lines in $\reals^2$.''
We remark that, in a subsequent work in progress, we present a truly subquadratic algorithm 
that solves Chan's problem in the algebraic decision-tree model,
with $O(n^{112/57+\eps})$ comparisons, for any $\eps>0$.

\section{Preliminaries}
\label{sec:prelim}

\paragraph*{Model of computation.}
First, we need the following strengthening of the RAM model.
We assume that this model supports root extractions, in the sense that for each $b\geq 1$,
the roots of a real univariate polynomial of degree $b$ can be computed in
time that depends only on $b$. This means that we obtain some discrete representation 
of the roots that allow us to perform comparisons and algebraic computations that 
involve these roots. See~\cite{AEZ-19} for a similar assumption.

\paragraph*{Polynomial partitioning.}
Our analysis relies on planar polynomial partitioning and on properties of 
Cartesian products of pairs of them.  For a polynomial 
$f \colon \reals^d \to \reals$, for any~$d \ge 2$, the \emph{zero set} of $f$ is
$Z(f) \coloneqq  \{ x\in\reals^d \mid f(x)=0 \}$. We refer to an open connected 
component of $\reals^d \setminus Z(f)$ as a \emph{cell}. 
The classical Guth-Katz result is: 

\begin{proposition}[Polynomial partitioning; Guth and Katz~\cite{GK}]
  \label{prop:GK}
  Let $P$ be a finite set of points in $\reals^d$, for any $d\ge 2$. 
  For any real parameter $D$ with $1 \le D \le \abs{P}^{1/d}$,
  there exists a real $d$-variate polynomial $f$ of degree $O(D)$ such that
  $\reals^d \setminus Z(f)$ has $O(D^d)$ cells, each containing at most 
  $\abs{P}/D^d$ points of $P$.
\end{proposition}

Agarwal, Matou{\v{s}}ek, and Sharir \cite{AMS-13} presented an algorithm that 
efficiently computes\footnote{%
  This polynomial forms a partition approximating the one in Proposition~\ref{prop:GK}, 
  and the constant of proportionality in the degree bound of~\cite{AMS-13} is somewhat larger.} 
such a polynomial~$f$, whose expected running time is $O(nr + r^3)$, where $r = D^d$.

Note that the number of points of $P$ on $Z(f)$ can be arbitrarily large.
For planar polynomial partitions, though, this can be handled fairly easily,
by partitioning the algebraic curve $Z(f)$ into subarcs,
each containing at most $|P|/D^2$ points (as do the complementary cells). 
We state this property formally and spell out the easy details in Appendix~\ref{app:hierarchical}.

\paragraph*{Polynomial partitioning for Cartesian products of point sets in the plane.}

Solymosi and De~Zeeuw~\cite{SdZ-18} studied polynomial partitioning for Cartesian
products of planar point sets.
Given two finite sets $P_1$ and $P_2$ of points in the plane, a natural strategy
to construct a partitioning polynomial for $P_1 \times P_2$ in $\reals^2\times\reals^2$, 
a space that we simply regard as $\reals^4$, is to construct
suitable bivariate partitioning polynomials $\varphi_1$ for $P_1$ and $\varphi_2$ for $P_2$,
as provided in Proposition~\ref{prop:GK}, and then take their product
$\varphi(x,y,z,w) \coloneqq \varphi_1(x,y) \varphi_2 (z,w)$.

\begin{corollary}[Polynomial partitioning of a Cartesian product~\cite{SdZ-18}]
  \label{cor:partition_CP}
  The partition of $P_{1,2} \coloneqq P_1\times P_2$ just described results in overall $O(D^4)$ relatively open cells of 
  dimensions $2$, $3$, and~$4$, each of which contains at most $\abs{P_{1,2}}/D^4$ points of $P_{1,2}$.
  The zero- and one-dimensional cells do not contain any point of $P_{1,2}$.
\end{corollary}

The analysis in~\cite{SdZ-18} also bounds the number of partition cells
intersected by a two-dimensional algebraic surface $S$ in $\reals^4$, provided
it has ``good fibers.''  We define this notion: 


\begin{definition}[Good fibers]
  (i)~A two-dimensional algebraic surface $S$ in $\reals^4 = \reals^2\times \reals^2$
  has \emph{good fibers} if, for every point $p \in \reals^2$, 
  with the possible exception of $O(1)$ points, 
  the fibers $(\{p\} \times \reals^2) \cap S$
  and $({\reals^2} \times \{p\}) \cap S$ are finite.
  (ii)~A two-dimensional algebraic surface $S$ in~$\reals^3 = \reals^2\times \reals$
  has \emph{good fibers} if, for every point $p \in \reals^2$, 
  with the possible exception of~$O(1)$~points, 
  the fiber $(\{p\} \times \reals) \cap S$ is finite (it is $\{p\}\times\reals$ 
  for each exceptional point $p$), and for every point $q \in \reals$, 
  with the possible exception of $O(1)$ points, the fiber $(\reals^2 \times \{q\}) \cap S$ 
  is a one-dimensional variety (i.e., an algebraic curve).
\end{definition}

Note that in this definition we are only concerned with one specific
decomposition of the underlying space into a product of two subspaces.

\begin{proposition}[Cells intersected by a surface~\cite{SdZ-18}]
  \label{prop:intersect}
  Let $S$ be a constant-degree two-dimensional algebraic surface in $\reals^4$ that has good fibers.
  Then $S$ intersects at most $O(D^2)$ two-, three-, and four-dimensional cells in
  the partitioning induced by $P_{1,2}$.
\end{proposition}

Both Corollary~\ref{cor:partition_CP} and Proposition~\ref{prop:intersect}
have (simpler) three-dimensional counterparts (i.e., in the context of a Cartesian product 
of a plane and a curve). These versions, while not stated explicitly, are addressed
in the analysis presented in Appendix~\ref{app:hierarchical}.

\section{Hierarchical Polynomial Partitioning}
\label{sec:hierarchical}

Even though we work in the algebraic decision-tree model, we still need to account
for the cost of constructing the various polynomial partitionings (as it requires
explicit access to the input points), which, if done by
a na\"ive application of the technique of \cite{AMS-13}, would be too expensive, as 
a direct implementation of our technique needs to use polynomials of high, 
non-constant degree. We circumvent this issue by constructing a
\emph{hierarchical polynomial partitioning}, akin to the
constructions of hierarchical cuttings of Chazelle~\cite{Chazelle-91} and
Matou\v{s}ek~\cite{Mat} from~the~1990s. The material is rather technical,
and, to make the presentation of our main results more accessible,
its details are delegated to Appendix~\ref{app:hierarchical}. 

Roughly, we gain efficiency by constructing a hierarchical tree of partitions 
using constant-degree polynomials, until we reach subsets of the input point set 
of the right size. Concretely, each node of the tree is associated with a cell
$\tau$ of some recursive partition, and with a subset~$P_\tau$ of points of $P$
that lie in\footnote{%
    A cell may contain additional points, but it is associated only with those points that 
    it contains and that are associated with its parent cell---see Appendix~\ref{app:hierarchical}.} 
$\tau$, and $P_\tau$ is partitioned recursively at the descendants of the node.

The actual hierarchical partitions that we will need are within a Cartesian product 
either of two planes or of a plane and a one-dimensional curve, and are obtained by 
taking suitable Cartesian products of partitions constructed within each of these subspaces.

We show that, up to $n^\eps$ factors, we achieve the same combinatorial properties 
as in a single-shot construction with a higher-degree polynomial, at a lower algorithmic cost.
Specifically, we have the following results.

\begin{theorem}[Cartesian product of two planar point sets]
  \label{thm:hier_partition_CP}
  Let $P_1$, $P_2$ be two sets of points in the plane, each of size $n$, 
  and put $P_{1,2} \coloneqq P_1 \times P_2\subset\reals^4 = \reals^2\times\reals^2$. 
  Let $1\le r \le n$ be an integer and $\eps>0$.\\
  (i) There is a hierarchical polynomial partition for $P_{1,2}$ with $O((n/r)^{2+\eps})$ bottom-level cells, 
  each of which is associated with a subset of at most $r^2$ points of~$P_{1,2}$, 
  which is the Cartesian product of a set of
  at most $r$~points from~$P_1$ and a set of at most $r$~points from~$P_2$, which it contains. 
  The number of such sets from $P_1$ (resp., $P_2$) is $O((n/r)^{1+\eps})$, each of which 
  is associated with a bottom-level cell in an appropriate hierarchical partition.
  The hierarchy can be constructed implicitly in expected $O(n\log n)$ time. \\
  (ii) Any constant-degree two-dimensional algebraic surface $S$ with good fibers
  reaches\footnote{%
    A surface $S$ is said to \emph{reach} a cell $\tau$ if it intersects $\tau$ 
    and all its ancestral cells---see Appendix~\ref{app:hierarchical}.} 
  at most $O((n/r)^{1+2\eps})$ cells at all levels of the hierarchical partition of $P_{1,2}$.
  These cells can be computed within the same asymptotic time bound $O((n/r)^{1+2\eps})$. \\
  (iii) For any algebraic curve $\gamma$ of degree at most $b$ (where $b$ is a constant), 
  the number of cells at all levels of the hierarchical partition of $P_{1,2}$ reached by 
  $\gamma$, and the time needed to find these cells, are $O((n/r)^{1/2+\eps})$, for any 
  $\eps > 0$. The constant of proportionality depends on $\eps$ and on $b$.
\end{theorem}

\begin{theorem}[Cartesian product of a planar point set and a 1D set] 
  \label{thm:hier_partition_12}
  Let $P$ be a set of $n$ points in the plane, and let $Q$ be a set of $n$ points lying on a
  constant-degree algebraic curve~$\gamma \subset \reals^2$.
  Let $1 \le r, s \le \sqrt{n}$ be integer parameters.\footnote{%
    The artificial threshold $\sqrt{n}$ is assumed because of a 
    certain technical step in the analysis---see Appendix~\ref{app:hierarchical} for details.}
  \\
  (i) There is a hierarchical polynomial partition for $P \times Q\subset \reals^2\times \gamma$
  into $O(n^{2+\eps}/(rs)^{1+\eps})$ bottom-level cells, for any $\eps > 0$,
  each of which is associated with a subset of at most $rs$ points of $P \times Q$,
  which is the Cartesian product of a set of at most $r$ points from $P$ and a set of at most $s$ points from $Q$.
  The number of such sets from~$P$ (resp., $Q$) is~$O((n/r)^{1+\eps})$ (resp.,~$O(n/s)$),
  each of which is associated with a bottom-level cell in an appropriate hierarchical partition.
  The hierarchy can be constructed implicitly in expected $O(n\log n)$ time. \\
  (ii) Any constant-degree two-dimensional surface $S$ with good fibers
  reaches (in the same sense as above) at most $O\left(\frac{n^{3/2+\eps}}{r^{1/2+\eps}s^{1+\eps}} \right)$
  cells at all levels of the hierarchical partition of $P \times Q$.
  These cells can be computed within the same asymptotic time bound
  $O\left(\frac{n^{3/2+\eps}}{r^{1/2+\eps}s^{1+\eps}} \right)$. 
\end{theorem}

\section{Collinearity Testing and Related Problems:\\ 
Vanishing of a Single Polynomial in $2{\times}1{\times}1$ Dimensions} 
\label{sec:col}

Let $A$, $B$, and $C$ be three sets of points in the plane, but assume that
$B$ and $C$ lie on respective constant-degree algebraic curves $\gamma_B$ and $\gamma_C$, 
assumed to be polynomially parameterized, meaning that they are given in
the parametric forms $\gamma_B(t) = (x_B(t),y_B(t))$ and $\gamma_C(s) = (x_C(s),y_C(s))$),
for $t,s\in \reals$, where $x_B(t)$, $y_B(t)$, $x_C(s)$, $y_C(s)$ are constant-degree polynomials.
With a small additional effort, the machinery developed here also applies to the case
where these are rational functions,\footnote{%
  The machinery also applies to situations where $x_B(t)$, $y_B(t)$, $x_C(s)$, $y_C(s)$ 
  are constant-degree continuous algebraic functions, but this involves additional
  technical complications, which we prefer to avoid in this presentation.}
but we stick to polynomiality for conciseness.
We make this assumption for all versions of the problems considered in this paper.

Our goal is to determine whether there exists a collinear triple of points
in $(a,b,c)\in A\times B\times C$, or more generally a triple $(a,b,c)$ satisfying some 
\emph{single} prescribed constant-degree polynomial equation $F(a,b,c) = 0$. 
For simplicity of exposition, we present our results in two stages.
We first focus on the collinearity testing problem, and simplify it further 
by assuming that $\gamma_C$ is a line. This simplifies some aspects of the analysis.
We then describe how to extend the analysis, for collinearity testing,
when $\gamma_C$ is a general constant-degree algebraic curve, and
for more general problems of this sort. 

In the first part of this section (Sections~\ref{sec:line}--\ref{sec:1poly})
we present our main algorithmic approach, 
which has the best complexity bound that we were able to obtain. 
In Section~\ref{sec:alternative} we present an alternative,
somewhat simpler, solution, which has an inferior 
(albeit still significantly subquadratic) complexity bound.

\subsection{Collinearity testing when $\gamma_C$ is a line} \label{sec:line}

In this subsection we consider the collinearity testing problem for the case
where $\gamma_C$ is a line (the case where the other curve $\gamma_B$ is a 
line is of course fully symmetric).

We assume that the sets $A$, $B$, and $C$ are pairwise disjoint, for otherwise 
collinear triples exist trivially; this condition can be checked efficiently.

In general, a triple $(a=(a_1,a_2),\;b=(x_B(t),y_B(t)),\;c=(x_C(s),y_C(s)))$ is collinear if and only if 
\[
\begin{vmatrix}
1 & a_1 & a_2 \\
1 & x_B(t) & y_B(t) \\
1 & x_C(s) & y_C(s)
\end{vmatrix}
 = 0 \,, \quad\text{or}
\]
\begin{equation}
  \label{eq:collinear}
x_B(t)y_C(s) - y_B(t)x_C(s) - a_1(y_C(s)-y_B(t)) + a_2(x_C(s)-x_B(t)) = 0 .
\end{equation}
In the special case where $\gamma_C$ is a line,
say the $x$-axis, we have $\gamma_C(s) = (s,0)$, and \eqref{eq:collinear}~becomes:
\begin{equation}
  \label{eq:colline}
- y_B(t)s + a_1y_B(t) + a_2(s-x_B(t)) = 0 ,  \text{\quad or\quad}
s = \varphi(a,t) \coloneqq \frac{a_1y_B(t) - a_2x_B(t)} {y_B(t) - a_2} .
\end{equation}
Here $\varphi$ is a constant-degree rational function; it is a linear rational function of $a$.

To simplify the analysis, we first dispose of some special situations, which can 
easily be detected and dealt with efficiently. Some of these assumptions reappear 
later when we present (in Section~\ref{sec:good}) a general technique for testing for
the good-fibers property, which we will need to do
for certain polynomials that arise in the analysis, as well as for other similar
properties, and for handling the problem in a simple efficient manner when these
properties do not hold (this section also handles, in this manner, some of the
properties that we need to assume).

We assume that $\gamma_B$ is an irreducible curve (for general collinearity testing, 
$\gamma_C$ is also assumed to be irreducible).
Otherwise, the analysis can be applied to individual irreducible components of these curves. 
(Since the curves are of constant degree, factoring each of them, over the real field,
into its irreducible components can be done in constant time under various assumptions
and model of computation; see, e.g., von zur Gathen~\cite{vzg} and Kaltofen~\cite{Kal}.)
If $\gamma_B$ and $\gamma_C$ are distinct, there are at most $O(1)$ points $q \in \gamma_B \cap \gamma_C$.
In this case, we will further assume that neither $B$ nor $C$ include such points.
Indeed, for any point $c \in \gamma_B \cap C$, we can detect any collinearity involving $c$ 
explicitly in $O(n \log n)$ time, by sorting points of~$A \cup B$ around~$c$ in angular order. 
If a collinearity is detected, we stop, otherwise we remove such points $c$ from $C$ and continue. 
Points of $B$ lying on $\gamma_C$ can be treated symmetrically. As there are only $O(1)$ 
such points, we have expended only $O(n \log n)$ work so far.

The analysis is not affected when $\gamma_B$ and $\gamma_C$ coincide but are not a line (in the general case).
If they coincide and are a common line, there is a collinearity if and only if $A$ has points
on that line, an easily detected scenario (we use here the assumption, already discussed above,
that $B$ and $C$ are disjoint).

To summarize, by all the assumptions made so far, $\gamma_C$ is the $x$-axis, $\gamma_B$ 
does not coincide with it, and no points of $B \cup C$ lie at the intersection of the two curves.
We may also assume that no points of $A$ lie on $\gamma_C$: any such point can participate in
a collinear triple only with points of $B\cap\gamma_C$, and we have just assumed that there are no such points.

Returning to the collinearity testing procedure,
we fix a pair of parameters $g,h\ll n$ (whose values will be set later)
and a parameter $\eps>0$,
and apply Theorem~\ref{thm:hier_partition_12}(i) to the sets $A$,~$B$,
with the respective parameters $r = g, s = h$.
Let $\tau$ (resp., $\tau'$) be a bottom-level cell in the resulting partition for $A$ 
(resp., $B$). Put $A_\tau \coloneqq A\cap\tau$ and 
$B_{\tau'} \coloneqq B\cap\tau'$. In this analysis, somewhat abusing the notation, 
we regard $B$ as a subset of $\reals$, and denote by $t$ the real parameter that
parameterizes~$\gamma_B$; in particular, we write $t \in B$ (resp., $t \in B_{\tau'}$) 
instead of $\gamma_B(t)\in B$ (resp., $\gamma_B(t) \in B_{\tau'}$). \footnote{%
  In Section~\ref{sec:1poly}, when we discuss the general problem, and also in
  Section~\ref{sec:good}, we use a different notation, denoting the set of the values
  of $t$ that parameterize the points of $B$ as $T$, with a similar replaced notation for $C$.}
The number of bottom-level cells~$\tau$ (and sets $A_\tau$) is 
$O\left(\left(\frac{n}{g}\right)^{1+\eps} \right)$, for any $\eps>0$, and
the number of bottom-level cells $\tau'$ (and sets $B_{\tau'}$) is $n/h$ 
(see also Appendix~\ref{app:hierarchical} for these details).

The high-level idea of the algorithm is to sort each of the sets
$\{ \varphi(a,t) \mid (a,t) \in A_\tau \times B_{\tau'} \}$,
over all pairs $(\tau,\tau')$ of cells, and then to
search with each $c=(s,0)\in C$ (i.e., with the corresponding real value~$s$)
through only the sorted lists that might contain $s$; 
the number of such lists is small, as argued below.

As in all works on this type of problems, starting from \cite{GP},
sorting the sets explicitly is too expensive, because the overall number of elements in the lists is quadratic.
Instead, we use, in the algebraic decision-tree model, a simple instance of the algebraic variant of
\emph{Fredman's trick}, extending the machinery used in the previous algorithms
for one-dimensional point sets~\cite{BCILOS-17,GP}.

\paragraph*{Preprocessing for batched point location.}

Consider the step of sorting the set $\{\varphi(a,t) \mid (a,t)\in A_\tau\times B_{\tau'}\}$,
for a pair of cells $\tau$, $\tau'$,
which has to perform various comparisons of pairs of values 
$\varphi(a,t)$ and $\varphi(a',t')$, for $a, a'\in A_\tau$, $t, t'\in B_{\tau'}$.
We perform this task globally over all pairs $(\tau,\tau')$ of cells.

We recurse by switching between the ``primal'' and ``dual'' setups.
In the primal, we view $P \coloneqq \bigcup_\tau A_\tau\times A_\tau$ as a set of 
$O\left(\left(\frac{n}{g}\right)^{1+\eps} \cdot g^2 \right) = O(n^{1+\eps} g^{1-\eps})$ 
points in $\reals^4$, and associate with each pair $(t,t') \in B_{\tau'}\times B_{\tau'}$, 
for each cell $\tau'$, the three-dimensional constant-degree algebraic surface 
\begin{equation} \label{eq:sigma}
\sigma_{t,t'} \coloneqq \{ (a,a') \in \reals^4 \mid \varphi(a,t) = \varphi(a',t') \} .
\end{equation}
We let $\Sigma$ be the collection of all these surfaces, over all %
cells $\tau'$, and have $\abs{\Sigma} = n/h \cdot h^2 = nh$.

In the dual, we view the pairs $(t,t') \in \bigcup_{\tau'} B_{\tau'} \times B_{\tau'}$
as points in the plane, and associate with each pair $(a,a') \in P$ the curve 
\begin{equation} \label{eq:delta}
\delta_{a,a'} \coloneqq \{ (t,t') \in \reals^2 \mid \varphi(a,t) = \varphi(a',t') \} .
\end{equation}
In each primal problem we need to perform batched point-location queries 
in an arrangement of (some subset of) the constant-degree algebraic 3-surfaces 
$\sigma_{t,t'}$ in $\reals^4$, and in each dual problem we need to perform 
batched point-location queries in an arrangement of (some subset of) the
constant-degree algebraic curves $\delta_{a,a'}$ in $\reals^2$. Initially we 
are in the primal, with $O(n^{1 + \eps} g^{1-\eps})$ points and $nh$ 3-surfaces.

If we could construct the full arrangement of these surfaces in the primal, or of these curves in the dual, and locate in it
all the respective points, we would obtain the signs of all the differences $\varphi(a,t) - \varphi(a',t')$,
for all $(a,t)$, $(a',t')\in A_\tau\times B_{\tau'}$, over all pairs $(\tau,\tau')$ of cells,
which would determine (at no extra cost in terms of comparisons) the sorted order of the sets
$\{\varphi(a,t) \mid (a,t)\in A_\tau\times B_{\tau'}\}$, over all pairs $(\tau,\tau')$.
However, a single-step construction of the full arrangement is too expensive, so we replace it with the above
``flip-flop'' primal-dual processing, each time partitioning (the current version of) the arrangement using 
a polynomial of small degree, and thereby reduce the cost to the subquadratic bound stated below.

The output of this preprocessing, obtained in a standard manner, is a representation
of $P\times\Sigma$ as a disjoint union $\bigcup_\alpha P_\alpha\times \Sigma_\alpha$ 
of complete bipartite graphs, where for each $\alpha$, the differences 
$\varphi(a,t) - \varphi(a',t')$, for all $(a,a')\in P_\alpha$, 
$\sigma_{(t,t')}\in \Sigma_\alpha$, have a fixed sign (similar approaches, in more 
involved forms, are also used in Sections~\ref{sec:1poly} and \ref{sec:fg}).  
For efficiency, each of these complete bipartite graphs is represented by the pair 
of its vertex sets, so the output size of the procedure is the sum of the sizes 
of the vertex sets of the graphs in the decomposition.

\begin{lemma}
  \label{lem:107}
  The above recursive batched point-location stage takes randomized
  expected time $O\left(n^{10/7+\eps'}g^{6/7+\eps'}h^{4/7}\right)$,
  where $\eps'$ is larger, by a small constant factor, than the prescribed~$\eps$.
  This also bounds the output size of the procedure, as just defined.
\end{lemma}
\begin{proof}
Put $M \coloneqq n^{1+\eps} g^{1-\eps}$, $N \coloneqq nh$, 
and fix two suitable sufficiently large constant parameters 
$r_1$, $r_2$, whose precise choice will be detailed shortly. 
We construct a $(1/r_1)$-cutting for~$\Sigma$, using vertical decomposition
(see~\cite{CEGS-91,Koltun-04} for details).
This forms a decomposition of $\reals^4$ into relatively-open pseudo-prisms of dimensions $0,\dots,4$, that we refer to
simply as \emph{prisms}, whose number is $O(r_1^{4+\eta})$, for any~$\eta>0$, so that each prism is crossed by
at most~$N/r_1$ surfaces of~$\Sigma$, and lower-dimensional prisms may be
fully contained in some of the other surfaces. By slicing the prisms into subprisms, 
we can also assume that each of these subprisms contains at most $M/r_1^4$ points of $P$;
somewhat abusing the terminology, we refer to the resulting subprisms also as prisms.

For each prism~$\zeta$ of the decomposition, we pass to the dual plane, 
with the set $P_\zeta^*$ of at most $M/r_1^4$ dual curves corresponding 
to points contained in~$\zeta$, and the set $\Sigma_\zeta^*$ 
of at most $N/r_1$ dual points corresponding to surfaces %
crossing~$\zeta$. We now apply a planar decomposition to this setting, 
using a $(1/r_2)$-cutting for the set of curves $P_\zeta^*$.
This results in $O(r_2^2)$ (pseudo-)trapezoids (see, e.g.,~\cite{CF-90}),
each crossed by at most $(M/r_1^4)/r_2 = \frac{M}{r_1^4r_2}$ dual curves,
and, after slicing the trapezoids into sub-trapezoids, as needed,
each trapezoid contains at most $(N/r_1)/r_2^2 = \frac{N}{r_1r_2^2}$ dual points.
Altogether, passing again to the primal, we end up with
$O(r_1^{4+\eta} r_2^2)$ subproblems, each involving at most $\frac{N}{r_1r_2^2}$~surfaces and at most $\frac{M}{r_1^4r_2}$ points.

Before processing each of these subproblems recursively, we form complete 
bipartite graphs, both in the primal and in the dual, 
where each such graph involves the points in some cell 
and the surfaces (or dual curves) that miss the cell or, when the cell is of lower dimension, 
fully contain the cell. Each graph is actually at most three different graphs, depending on the (fixed) sign
of the corresponding differences $\varphi(a,t)-\varphi(a',t')$, of the points in the cell, with respect
to the missing or containing surfaces or dual curves (recall the equations of these surfaces and curves in~\eqref{eq:sigma} and~\eqref{eq:delta}).
The overall collection of these graphs, including the trivial ones 
constructed at the bottom of the recursion, constitutes the output of this preprocessing.

We proceed recursively, alternating between primal and dual spaces, for $j$ levels, reaching
a total of at most $c^jr_1^{(4+\eta)j} r_2^{2j}$ subproblems, for a suitable constant $c$,
each involving at most $\frac{N}{r_1^j r_2^{2j}}$ surfaces and at most $\frac{M}{r_1^{4j} r_2^j}$ points.
We choose $j$, $r_1$ and $r_2$ so as to satisfy
\begin{equation}
  \label{eq:r1_r2}
  r_1^jr_2^{2j}  = N \quad \text{and} \quad
  r_1^{4j}r_2^j  = M ,
\end{equation}
or
\[
r_1^j  = \frac{M^{2/7}}{N^{1/7}} \quad \text{and} \quad
r_2^j  = \frac{N^{4/7}}{M^{1/7}} .
\]
This choice is valid when both expressions $M^{2/7}/N^{1/7}$ and $N^{4/7}/M^{1/7}$ are at least $1$.
That is, we require that $M^2\ge N$ and $N^4 \ge M$, or that 
\[
  h \le n^{1+2\eps}g^{2-2\eps} \qquad\text{and} \qquad g^{1-\eps} \le n^{3-\eps}h^4 .
\]
As both inequalities trivially hold, it follows that we can choose $r_1$ and $r_2$ (and $j$) to satisfy the above equalities.

The conditions in~\eqref{eq:r1_r2} make the bottom-level subproblems have constant size and
leads to the overall running time and storage size
\[
O\left(c^j(r_1^j)^{4+\eta} (r_2^j)^{2}\right) 
= O\left(c^j M^{6/7+2\eta/7}N^{4/7-\eta/7} \right) 
= O\left( n^{10/7+\eps'}g^{6/7+\eps'}h^{4/7} \right) ,
\]
where $\eps'$ depends on $\eps$ and $\eta$ (and on $c$); with a suitable choice of $\eta$
it is only slightly larger than our prescribed $\eps$.
\end{proof}

\paragraph*{Searching with the points of $C$.}
We next search the structure with every $s\in C$ (identified with the point $(s,0)$ on the $x$-axis).
For each $s\in C$, we only want to visit subproblems $(\tau,\tau')$ where there 
might exist $a\in\tau$ and $t\in\tau'$ (not necessarily from $A_\tau\times B_{\tau'}$), such that $\varphi(a,t)=s$.
We consider the two-dimensional surface 
\[
\pi_s \coloneqq \{(a,t)\in \reals^3 \mid \varphi(a,t)=s \} .
\]

To apply the machinery laid out in Sections~\ref{sec:prelim} and~\ref{sec:hierarchical}, 
we need to verify that the surfaces $\pi_s$ have the good fibers property, with the 
possible exclusion of $O(1)$ exceptional values of $s$. Rather than presenting ad-hoc 
arguments for the special setup at hand (as we did in an earlier version~\cite{AES:prev}),
we present, in Proposition~\ref{thm:2x1x1-one} below, 
a general efficient mechanism that either finds explicitly a zero of $F$ on $A\times B\times C$, 
or certifies that no such zero exists,
or else guarantees that the good fibers property holds for the polynomials $\pi_s$.
We therefore continue under the assumption that these polynomials do have the good fibers property.

By Theorem~\ref{thm:hier_partition_12}(ii), choosing $g$ and $h$ to satisfy
$\left(\frac{n}{g}\right)^{1/2} = \frac{n}{h}$, 
or $h = n^{1/2}g^{1/2}$, we ensure that $\pi_{s}$ reaches
$O\left( \frac{n^{1+\eps}}{g^{1+\eps}} \right)$ cells $\tau\times\tau'$.
Summing over all the $n$ possible values of $s$, the number 
of crossings between the surfaces $\pi_s$ and the cells $\tau\times\tau'$ 
is $O\left(\frac{n^{2+\eps}}{g^{1+\eps}}\right)$.
Denoting by $n_{\tau,\tau'}$ the number of surfaces
$\pi_s$ that cross $\tau\times \tau'$, we have
${\displaystyle \sum_{\tau,\tau'} n_{\tau,\tau'} = O\left(\frac{n^{2+\eps}}{g^{1+\eps}}\right)}$ and
we can enumerate all such crossings in~$O(n^{2+\eps}/g^{1+\eps})$~time.

The cost of searching with any specific $s$ in the structure of a cell $\tau\times\tau'$ crossed by~$\pi_s$,
is~$O(\log g)$ (it is simply a binary search over the sorted list of the values $\varphi(a,t)$ 
in each such cell, where these lists are prepared free of charge,
in the algebraic decision-tree model, from the complete bipartite graph representation
obtained at the preprocessing point-location stage).
Hence the overall cost of searching with the elements 
of $C$ through the structure is~$O(n^{2+\eps}/g^{1+\eps})$, where $\eps$ is slightly
larger than the originally prescribed one. 

Combining this cost with that of the construction of the hierarchical polynomial partitioning,
and the point-location preprocessing stage, we get overall expected time of 
\[
  O\left(n\log{n} + n^{10/7+\eps}g^{6/7+\eps}h^{4/7} + \frac{n^{2+\eps}}{g^{1+\eps}} \right) =
  O\left(n\log{n} + n^{12/7+\eps}g^{8/7+\eps} + \frac{n^{2+\eps}}{g^{1+\eps}} \right) .
\]
We roughly balance the two last terms by choosing $g = n^{2/15}$,
making the overall cost of the procedure
${\displaystyle O\left( \frac{n^{2+\eps}}{g^{1+\eps}} \right) = O\left( n^{28/15+\eps} \right)}$.
Again, $\eps$ here is slightly larger than the prescribed value.
 
\subsection{Testing and ensuring the good-fibers property}
\label{sec:good}

To complete the analysis, we next present a general technique for
testing and ensuring the good-fibers property.
This technique will also be used in the solution of the general single polynomial 
vanishing problem, presented in the next subsection.

\begin{proposition} \label{thm:2x1x1-one}
Let $F(x,y,z)$ (with $x,y,z\in \reals^2$) be a real 6-variate polynomial. 
Let $\gamma_B(t)$, $\gamma_C(s)$ be polynomial or rational parametric representations of 
respective algebraic curves $\gamma_B,\gamma_C$ in the plane.  
Let $A$ be a set of $n$~points in the plane, and let $T$ and $S$ 
be two sets of $n$ real parameters, with corresponding point sets 
$B\coloneqq\gamma_B(T)$ on~$\gamma_B$ and $C\coloneqq\gamma_C(S)$ on~$\gamma_C$.

\noindent
Then, after some preprocessing of $F$, $\gamma_B$, and~$\gamma_C$, 
we can transform the problem, in~$O(n\log n)$ time, into one where 
$H(x,t,s)\coloneqq F(x,\gamma_B(t),\gamma_C(s))$ satisfies the good fibers property, 
for every value of $s$, excluding $O(1)$ exceptional values (which we process
separately), as formulated in Section~\ref{sec:prelim}, 
or detect a triple~$(a,b,c) \in A\times B\times C$ 
such that $F(a,b,c)=0$, or conclude that no such triple exists.
\end{proposition}

We prove Proposition~\ref{thm:2x1x1-one} in several stages, by a sequence of reductions, where
each reduction, when implemented, either (i) detects a triple~$(a,b,c) \in A\times B\times C$ 
such that $F(a,b,c)=0$, or (ii) guarantees that no such triple exists,
or (iii) identifies $O(1)$ exceptional points in one of the sets, which we process separately
or put aside before applying further reductions.
Consider, as in the proposition statement, the 4-variate real polynomial $H(x,t,s) \coloneqq F(x,\gamma_B(t),\gamma_C(s))$.  
Henceforth, we focus on the polynomial $H$, looking for a triple~$(a,t,s)\in A \times T \times S$ with $H(a,t,s)=0$.
Each reduction will allow us to make a simplifying assumption on~$H$ in order to enforce the good fibers condition,
or to abort the process, either positively, by directly determining the presence of a 
desired triple $(a,t,s)\in Z(H)$, or negatively, by determining that no such triple exists.

First, observe that we can safely assume that $H$ is not the zero polynomial, 
as otherwise every triple $(x,t,s)$ qualifies and no comparisons are needed.
In addition, we make the following sequence of assumptions, justifying each of them as we go.
\begin{assumption}[Irreducibility]
$H$ is an irreducible polynomial.
\end{assumption}
%
Since $H$ is a given constant-degree polynomial that does not depend on the input sets~$A$,~$B$,~$C$, testing for its irreducibility, and factoring it when it is reducible, 
are tasks that do not depend on $A$, $B$ and $C$, and can therefore be ignored in the 
algebraic decision-tree model that we use. 
The actual algorithmic issues in testing for irreducibility and of factoring multivariate
polynomials have been covered in several works, such as \cite{Kal,vzg}. The applicability
of these algorithms depends on the model of computation and on certain assumptions on the 
polynomials. For simplicity of presentation, we finesse these issues and assume, as stated,
that $H$ is irreducible. 


We now recall and extend the good fibers condition, expanded and recast in our present context:

\begin{definition}[Good fibers condition for $H(x,t,s)$]
  For each $s_0\in S$, excluding $O(1)$ {exceptional values}, both of the following conditions hold:
  \begin{enumerate}
  \item For any $a_0\in\reals^2$ (not necessarily from $A$),
    excluding $O(1)$ exceptional values, the equation
    $H(a_0,t,s_0) = 0$ in the variable $t$ has $O(1)$ solutions.
  \item For any $t_0\in\reals$ (not necessarily from $T$),
    excluding $O(1)$ exceptional values, the equation
    $H(x,t_0,s_0)=0$ in $x$ defines a set in $\reals^2$ whose dimension is at most $1$.
  \end{enumerate}
\end{definition}
Note that this definition differs from the one given in Section~\ref{sec:prelim},
in that it adds the variable $s$ as another independent variable, and formulates
the good fibers property as a uniform condition that holds for all but 
$O(1)$ exceptional values of $s$.

We first argue that it is safe to make the following non-degeneracy assumption:
\begin{assumption}[Non-degeneracy]
  $H(x,t,s)$ depends on all three variables, in the sense that it cannot be 
  written as a polynomial in a proper subset of these variables.
\end{assumption}

To see how to test for and enforce this assumption, suppose that $H(x,t,s)$ 
does not depend on $s$, i.e., $H(x,t,s)=h(x,t)$ for some real trivariate polynomial 
$h(x,t)$, for all $x$, $t$, and~$s$ (recall that $x$ is a point in the plane). 
Then the problem reduces to finding $a \in A$, $t \in T$ so that $h(a,t)=0$.  
For each $a \in A$ for which $h(a,t)$ is not identically zero (as a polynomial in $t$), 
we compute the $O(1)$ roots $t_{a,0}$, $t_{a,1}, \ldots$ of $h(a,t)=0$, in $O(1)$ time
(recall the model of computation that we have assumed in Section~\ref{sec:prelim}),
and collect them into a set~$T'$. We then check, in $O(n\log n)$ time, whether the overall set $T'$ and $T$ 
have any common elements, thereby detecting the existence of a pair $(a,t) \in A \times T$ 
with $h(a,t)=0$ and solving the problem (appending any real number $s$ to $(a,t)$
gives us a triple satisfying $H(a,t,s)=0$). There may also be $O(1)$ values of $a$ for which
$h(a,t)\equiv 0$. If such an $a$ belongs to $A$, any pair $(t,s)$ of real numbers
gives us a triple $(a,t,s)$ at which $H$ vanishes. If neither of the two kinds of tests
succeeds, we conclude that $H(x,t,s)=h(x,t)$ has no zeros on $A\times T\times S$.

The case of $H$ not depending on $t$ is handled symmetrically.  
If $H$ does not depend on $x$, a similar and slightly easier argument suffices, 
thereby completing the reduction.  

\medskip
We now write $H$ as
\[
  H(x,t,s) = \sum_{i,j,k} h_{ijk}(s)x_1^i x_2^j t^k ,
\]
for suitable polynomial coefficients $h_{ijk}$.
For a fixed $s_0$, $H(x,t,s_0)$ is identically zero as a polynomial in $x$ and $t$ 
if and only if $H^*(s)\coloneqq \sum_{i,j,k} h_{ijk}^2(s)$ vanishes at $s = s_0$. 
$H^*(s)$ is a never-negative univariate real polynomial; it is not the zero polynomial, 
for otherwise $H$ would have been the zero polynomial too. Hence $H^*(s)$ vanishes only 
at a finite (and constant) number of values of $s$. If such a value $s=s_0$ exists, every $h_{ijk}$ 
vanishes at $s_0$, so every~$h_{ijk}(s)$ is divisible by $s-s_0$,
and therefore $H$ is divisible by $s-s_0$. This violates either the irreducibility 
or the non-degeneracy assumption, depending on whether $H(x,t,s)$ is a constant or 
a non-constant multiple of $s-s_0$.

Therefore we henceforth assume that $H(x,t,s_0)$ is not the zero polynomial in 
$x$ and $t$, for any value $s_0$.
%

\begin{proof}[Proof of Proposition~\ref{thm:2x1x1-one}]
Having made and justified these assumptions, we now proceed to the proof of the proposition.

\medskip
\noindent\textbf{Condition (2).}
Suppose the solution set of $H(x,t_0,s_0)=0$, for some choice of real numbers~$t_0,s_0$, is of dimension larger than one. 
Since $H(x,t_0,s_0)$ is a real polynomial in~$x$, it must then be the zero polynomial.
Writing $H(x,t,s)=\sum_{ij}h_{ij}(t,s) x_1^ix_2^j$, for suitable polynomial coefficients 
$h_{ij}(t,s)$, we see that $H(x,t_0,s_0)$ is identically zero as a function of~$x$ 
precisely when $h_{ij}(t_0,s_0)=0$, for all $i,j$.
If the solution to this system of equations (in which we replace the fixed values 
$t_0$, $s_0$ by variables $t$, $s$) is a discrete set, then there are only 
a constant number of exceptional pairs $(b,c)$, parameterized by the roots $(t_0,s_0)$,
for which the solution set of $H(x,t_0,s_0)$ has dimension higher than one, 
and Condition~(2) is satisfied. Otherwise, by B\'ezout's theorem, 
all the polynomials $h_{ij}$ have a nontrivial common factor~$h(t,s)$.
(It is impossible for all $h_{ij}$ to be identically zero, as so 
would be~$H$.) That is, we have $H(x,t,s)=h(t,s)f(x,t,s)$, for some suitable
polynomial $f$, which violates irreducibility, 
unless $f$ is a constant polynomial and then it violates the non-degeneracy assumption.
This finishes the argument for the solutions of $H(x,t_0,s_0)=0$, namely for Condition~(2) of the good fibers definition.

\medskip
\noindent\textbf{Condition (1).}
Now consider the equation $H(a_0,t,s_0)=0$. 
Suppose $H(a_0,t,s_0)$ is identically zero, for some pair~$(a_0,s_0)$, as a polynomial in~$t$.
This means that the surface~$Z(H)$ contains a line
parallel to the $t$-axis. To see when this can happen, write $H(x,t,s) = \sum_i g_i(x,s) t^i$, 
for suitable polynomial coefficients $g_i(x,s)$, and define $G(x,s)\coloneqq\sum_i g_i^2(x,s)$ 
($G$ cannot be the zero polynomial, as otherwise $H$ would be too). 
The line $x=a_0,s=s_0$ is contained in $Z(H)$ iff all the~$g_i$'s are zero at $(a_0,s_0)$, or, 
equivalently, if $(a_0,s_0)\in Z(G)$.

Since $G$ is not the zero polynomial, $Z(G)$ cannot be the entire $as$-space.
If $Z(G)$ has a two-dimensional irreducible component, the product of
this component with the $t$-axis in the $ats$-space (which is $\reals^4$) 
would be (by definition of $G$) contained in~$Z(H)$,
and thus consist of one or several irreducible components of $Z(H)$. 
This in turn would either violate the irreducibility assumption, 
or show that $H$ does not depend on $t$, contradicting non-degeneracy. 

We can therefore assume that $\dim Z(G)<2$.

If $\dim Z(G)=0$, there is only a finite list of pairs $(a_0,s_0)$ 
for which this happens; these are the exceptional points and Condition~(1) is satisfied.

Suppose then that $\dim Z(G)=1$. The set $Z(G)$ either contains a line $\ell$ 
parallel to the $s$-axis as one of its irreducible components or it does not.

If it does, the product of $\ell$ with the $t$-axis is a 2-plane of the form 
$(x_1,x_2) = (a_1,a_2)$ contained in $Z(H)$. That is, the
point $a=(a_1,a_2)$ (not necessarily from $A$) is such that
$H(a,t,s)$ is identically zero in $t$, $s$. There is only a constant number 
of such lines (at most the degree of $G$); they correspond to exceptional 
points~$a$ in Condition~(1) of the good fiber definition. 
Algorithmically, there is a finite number of points $a \in A$ 
for which we do not need to look at $t_0$ or $s_0$ to check whether $H(a,t_0,s_0)=0$. 
Again, since $H$ is given as part of the problem statement, this means just
checking the points of~$A$ against a precomputed list (of length~$O(1)$),
of these exceptional points, obtained by comparing the coefficients of $G$ 
to zero, which takes linear time. This yields a constant number of exceptional 
pairs $(a,c)$ parameterized by the corresponding roots.

Now remove from $Z(G)$ all irreducible components that are lines parallel to the $s$-axis and, 
slightly abusing the notation, still refer to the remaining set as $Z(G)$; 
if it is empty, we are done, so assume otherwise. That means that the 
intersection of $Z(G)$ with any line of the form~$x={\rm const}$ (in the $xs$-space) is finite.  
In other words, for any fixed $a$, there are only $O(1)$~values~$s_0$, such that $H(a,t,s_0)=0$ 
is identically zero as a function of~$t$; these are the exceptional points
($(a,c)$, where $c$ is parameterized by $s_0$) in Condition~(1) of the good fibers definition. 
For all other values of $a$, $H(a,t,s_0)=0$ is not 
identically zero and therefore has only $O(1)$ solutions, as desired.

To complete the analysis, we need to check, for each exceptional value $c_0$ in $C$, 
whether the corresponding polynomial $H(x,t,s_0)$ (where $s_0$ is the parameter value for $c_0$)
vanishes on~$A\times T$. This can be done exactly as in the argument following the 
non-degeneracy assumption above, and takes $O(n\log n)$ time.

This completes the proof of Proposition~\ref{thm:2x1x1-one}.
\end{proof}

To summarize, unless the machinery presented in this subsection produces a zero of~$F$ on~$A\times B\times C$, 
or asserts that no such zero exists, we have the good fibers property for every~$s$ 
(that is,~$c$) in~$C$, excluding $O(1)$ exceptional values, and the analysis can proceed, as detailed above.
%

\subsection{The general case}
\label{sec:1poly}

A similar analysis, albeit somewhat more complicated, can handle the case 
where $C$ is contained in an arbitrary constant-degree algebraic curve (given in parametric form as above), 
rather than a line, and the general case, where we want to test for the vanishing of a single
arbitrary constant-degree irreducible polynomial $F$ on $A\times B\times C$. More precisely, we test
for the vanishing of $H(x,t,s) \coloneqq F(x,\gamma_B(t),\gamma_C(s))$ on $A\times T\times S$, where
$T$ and $S$ are the sets of real numbers that represent $B$ and $C$, respectively.

We follow the notation and analysis in Section~\ref{sec:line}.
Most of the analysis carries over, but certain nontrivial technical modifications are needed.
A major technical issue, however, is the need to modify the definitions of the
surfaces $\sigma_{t,t'}$ and the curves $\delta_{a,a'}$ (see the respective 
equations~\eqref{eq:sigma} and~\eqref{eq:delta}), and the way we manipulate them,
which is needed because, in general, the equation $H(x,t,s)=0$ no longer admits an explicit 
solution of the form $s = \varphi(x,t)$ that we had in the preceding case, so the equation
may have several $s$-roots for given $x$ and $t$. The analysis in~\cite{BCILOS-17} had to 
face a similar challenge (for the one-dimensional case), and our analysis is based in part 
on the ideas developed there, extended to our setting.

As a consequence, for a point $a_0\in\reals^2$ and a number $t_0\in\reals$, there are 
in general finitely many (but more than one) roots of the equation $H(a_0,t_0,s) = 0$. 
(We treat the case where the equation has infinitely many roots below.)
The idea is to collect all these roots, over all pairs $(a_0,t_0)\in A\times T$, sort them
into a list $\Lambda$, and search with each $s\in S$ in $\Lambda$. If any such~$s$ is 
found to belong to $\Lambda$, we have found a triple in $A\times B\times C$ 
on which $F$ vanishes.  In addition, we need to check whether $A\times T$ contains
any pair $(a_0,t_0)$ for which $H(a_0,t_0,s) = 0$ has infinitely many $s$-roots 
(i.e., $H(a_0,t_0,s)$ is the identically zero polynomial in $s$), in which case
$H$ vanishes at every triple $(a_0,t_0,s)$, for $s\in S$. If none of these steps
succeeds, $H$ does not vanish on $A\times T\times S$. We denote the sorted sequence 
of the $s$-roots of the equation $H(a,t,s) = 0$, for a fixed pair $(a,t)$ for 
which the equation is not identically zero, as $\Xi(a,t)$.

We first dispose of criticalities that occur at pairs $(a,t)$ for which $H(a,t,s)$ is identically 
zero as a polynomial in $s$. Writing $H(a,t,s) = \sum_{i\ge 0} h_i(a,t)s^i$, this 
happens when all the polynomials~$h_i(a,t)$ vanish simultaneously. The locus~$\psi$ of this criticality is a variety in the $at$-space (that is, in~$\reals^3$)
of dimension at most two (otherwise all $h_i$ are identically zero and therefore so is $H$, a case that is easy to detect). 
Recall the definition of having good fibers, as applied to $\psi$ (refer to Section~\ref{sec:prelim}). 
If $\psi$ has this property, for each $a\in A$ (with~$O(1)$~exceptions) there are 
only $O(1)$ values $t$ such that $(a,t)\in\psi$ (we denote this set of points as $\hat{\eta}_a$), 
and for each $t\in T$ (with~$O(1)$~exceptions) the points~$a$ that satisfy $(a,t)\in\psi$ lie 
on a curve, denoted as $\hat{\zeta}_t$. It is easy (and inexpensive) to test, in the spirit of 
Section~\ref{sec:good}, whether exceptions exist, or whether $\psi$ fails to satisfy the 
good-fibers property. If any of these cases is detected,
we find either a point $a\in A$ such that $H(a,t,s)$ is identically zero as 
a polynomial in $t$ and $s$ (because all the coefficients $h_i(a,t)$ are 
identically zero for all $t$), or a point $t\in T$ such that $H(a,t,s)$ is 
identically zero as a polynomial in $a$ and $s$. In either case we get 
(many) triples of $A\times T\times S$ at which $H$ vanishes.

We may thus assume that $\psi$ has good fibers and disregard the exceptions, if they exist. 
As in Section~\ref{sec:line}, we traverse all pairs of bottom-level cells $(\tau,\tau')$,
in the joint hierarchical partition of~$A$ and~$T$
that are crossed by $\psi$. In each of these cells we check,
by brute force, for all pairs $(a,t)$ whether they lie on $\psi$. If such a pair 
$(a_0, t_0)$ is found, $H$ vanishes at all the triples $(a_0, t_0, s)$, for
any real number $s$, so we report that there exists a vanishing triple. Otherwise, 
we continue the algorithm under the assumptions that for any point $a_0\in A$ 
and any number $t_0\in T$, there are only finitely many $s$-roots of $H(a_0,t_0,s)$.
Using the parameters $g$, $h$ of Section~\ref{sec:line}, and recalling 
Theorem~\ref{thm:hier_partition_12}(ii), the number of cells~$\tau\times\tau'$
crossed by~$\psi$ is~$O(n^{1+\eps}/g^{1+\eps})$, and they can be found
within a comparable time bound. Since each such cell contains $gh = n^{1/2}g^{3/2}$
points, the overall cost of this step is~$O(n^{3/2+\eps}g^{1/2-\eps})$, well 
below the overall time bound of the entire algorithm.

We thus proceed as follows.
In order to sort the list $\Lambda$, we need to know the outcome of comparisons between pairs of its 
elements. If the pair involves two roots of distinct sequences $\Xi(a,t)$, $\Xi(a',t')$, 
then the roots coincide when $(a,a')$ lies on the surface $\sigma^{(0)}_{t,t'}$, 
or, in an equivalent dual context, when $(t,t')$ lies on the dual curve $\delta^{(0)}_{a,a'}$, where
\begin{align*}
\sigma^{(0)}_{t,t'} &\coloneqq \{ (a,a') \in\reals^4 \mid \exists s\in\reals : H(a,t,s) = H(a',t',s) = 0 \} \\
\shortintertext{and}
\delta^{(0)}_{a,a'} &\coloneqq \{ (t,t') \in\reals^2 \mid \exists s\in\reals : H(a,t,s) = H(a',t',s) = 0 \} .
\end{align*}
By eliminating $s$ from these expressions (see, e.g.,~\cite{CLO} for details), we 
replace them by more involved semi-algebraic formulas that depend only on $a$, $a'$, $t$, and $t'$.
As a matter of fact, $\sigma^{(0)}_{t,t'}$ is defined by the equation 
$\res(H(a, t, s), H(a', t', s); s) = 0$ (where $\res(\cdot,\cdot;s)$ denotes the resultant with respect to~$s$), 
which is a polynomial in $a$ and $a'$ (for the fixed pair $t,t'$), whose degree 
depends on the degree of $H$, and which can be computed in constant time (see, e.g.,~\cite{CLO}).

If the pair involves two roots of the same sequence $\Xi(a,t)$, then the roots coincide
(that is, become a double root) when $a$ lies on the curve $\zeta_t$, 
or, equivalently, when $t$ belongs to the set $\eta_a$, where
\begin{align*}
  \zeta_t &\coloneqq \{ a\in\reals^2 \mid \exists s\in\reals : H(a,t,s) = \frac{\bd H}{\bd s}(a,t,s) = 0 \} \\
\shortintertext{and}
  \eta_a &\coloneqq \{ t\in\reals \mid \exists s\in\reals : H(a,t,s) = \frac{\bd H}{\bd s}(a,t,s) = 0 \} .
\end{align*}
That is, if we fix $t$ and consider the two-dimensional surface $H_t \coloneqq  \{(a,s) \mid H(a,t,s) = 0\}$,
within the three-dimensional $as$-space, then the $a$-coordinates of the points on $H_t$, 
that are either singular or have an $s$-vertical tangent line, are those that comprise the curve $\zeta_t$.
Similarly, if we fix $a$ and consider the curve $H_a \coloneqq \{(t,s) \mid H(a,t,s) = 0\}$
in the $ts$-plane, then the $t$-coordinates 
of the points on $H_a$ that are either singular or have an $s$-vertical tangency, comprise the set $\eta_a$.
Here too, the computation of $\zeta_t$ and $\eta_a$ depends on the degree of $H$, and thus takes constant time.



The high-level view of the analysis is as follows. The surface $\sigma^{(0)}_{t,t'}$ is the
locus of all $(a,a')$ for which a root of $\Xi(a,t)$ coincides with a root of $\Xi(a',t')$.
There are many such possible coincidences, between the $i$th root of one sequence and
the $j$th root of the other one, for various pairs $i,j$ of indices. To keep control of
these coincidences, we do two things. First, we add to $\sigma^{(0)}_{t,t'}$ the
pair of surfaces $\zeta_t \times \reals^2$ and $\reals^2\times \zeta_{t'}$, and denote
the new surface as $\sigma_{t,t'}$. Second, we decompose $\reals^4$ into the faces,
of various dimensions, of the arrangement formed by $\sigma_{t,t'}$. 
A formal, algebraically well-defined procedure for obtaining (a refinement of) this 
decomposition is to construct the \emph{cylindrical algebraic decomposition (CAD)} 
of $\sigma_{t,t'}$ (see \cite{Col,SS2} for details), and take the cells of the CAD 
(also known as \emph{strata}) to form the desired decomposition of space.

One can show that, for each stratum $\Delta$, each of the sorted sequences $\Xi(a,t)$, 
$\Xi(a',t')$, as well as their merged sequence, which we denote as $\Lambda(a,t;a',t')$
has a fixed combinatorial structure, over all points $(a,a') \in \Delta$.
Roughly speaking, this follows from the fact that we have coinciding $s$-roots over 
the surfaces $\sigma_{t,t'}$, and the points of singularity of such a surface
(when $\sigma_{t,t'}$ intersects itself) are the loci of points with multiple 
repeating $s$-roots. The CAD of $\sigma_{t,t'}$ captures all such types of singularity, 
and therefore in each stratum $\Xi(a,t)$, $\Xi(a',t')$, and $\Lambda(a,t;a',t')$ are invariant.
That is, the following properties hold.

\medskip
\noindent (1)
The number $k_{a,t}$ of distinct roots in~$\Xi(a,t)$ (resp., the number $k_{a',t'}$ 
of distinct roots in~$\Xi(a',t')$) remains the same, over $(a,a')\in\Delta$, and 
the root order within each sequence does not change. In particular, if one or both 
sequences have roots with multiplicity, which can happen when $\Delta$ is of dimension 
three or smaller, each multiple root retains its multiplicity over all $(a,a') \in \Delta$.
We also have the property that each root varies continuously as $(a,a')$ varies in $\Delta$.\footnote{%
  We do not give a formal proof of this property, but, e.g., for four-dimensional cells 
  it is a consequence of the implicit function theorem (see, e.g.,~\cite{BCR-98}), and 
  follows by an easy continuity-based argument when $\xi$ is of lower dimension.}

\medskip
\noindent (2)
The set of pairs of indices $(i,j)$, for $i=1,\ldots, k_{a,t}$, $j=1,\ldots, k_{a',t'}$,
so that the $i$th root of $\Xi(a,t)$ coincides with the $j$th root of $\Xi(a',t')$, 
is fixed throughout $\Delta$. This implies that, for each pair of indices $i$, $j$,
the relative order (i.e., larger/smaller/equal) between the $i$th root of $\Xi(a,t)$ 
and the $j$th root of $\Xi(a',t')$, is fixed throughout $\Delta$.

\medskip
Given these properties, we can explicitly compute the combinatorial structure of the 
sorted sequences $\Xi(a,t)$, $\Xi(a',t')$, and of their merged sequence $\Lambda(a,t;a',t')$,
for any $(a,a') \in \reals^4$. To do so, we pick an arbitrary point $(a_0, a_0')$ inside 
each stratum $\Delta$ of $\sigma_{t,t'}$, and then extract the $s$-roots of $H(a_0,t,s)$
and of $H(a_0',t',s)$, sort the resulting sequences and then merge them together. 
This gives us the desired fixed combinatorial structure of the three resulting sequences,
over all $(a,a') \in \Delta$. This operation can be performed in overall constant time 
(which depends on the degree of~$H$), using properties of the CAD and our model of 
computation---see Section~\ref{sec:prelim}.

A symmetric (and simpler) situation occurs in the dual. That is, we take each curve
$\delta^{(0)}_{a,a'}$, add to it the lines in $\eta_a \times \reals$ and 
$\reals\times \eta_{a'}$, call the resulting curve $\delta_{a,a'}$, construct
its CAD (in the plane), and denote the cells of the CAD as the \emph{strata} of the
desired planar decomposition. Then properties~(1) and~(2) hold in this dual
context too, now for the fixed pair $(a,a')$, and for all $(t,t')$ in any of
the strata $\Delta$ of that decomposition.

The preceding arguments are summarized in the following lemma.

\begin{lemma} \label{inv-sigma-delta}
Let $\Delta$ be a stratum, that is, a cell, of any dimension, 
of the CAD formed by the surface $\sigma_{t,t'}$. 
Then, as $(a,a')$ varies continuously within $\Delta$, the following properties hold. \\
(i) Each of the sequences $\Xi(a,t)$ and $\Xi(a',t')$ has a fixed combinatorial structure. \\
(ii) The sorted merged sequence $\Lambda(a,t;a',t')$ also has a fixed 
combinatorial structure, and any coincidence of roots of one sequence with roots of the 
other sequence does not change. More generally, for each $i=1,\ldots,k_{a,t}$, 
$j=1,\ldots,k_{a',t'}$, the relative order of the $i$\emph{th} root of $\Xi(a,t)$ and 
the $j$\emph{th} root of $\Xi(a',t')$ (including possible equality between them) remains invariant. \\
(iii) Each of the roots of either sequence $\Xi(a,t)$, $\Xi(a',t')$, or a common 
root of both in $\Lambda(a,t;a',t')$, varies continuously with $(a,a')\in\Delta$.

\medskip
\noindent
Symmetrically, if $(t,t')$ varies continuously within a dual stratum $\Delta$, of any dimension, 
of the decomposition formed by the curve $\delta_{a,a'}$, the
properties~(i)--(iii) hold for~$(a,t)$ and~$(a',t')$, with obvious modifications.
\end{lemma} 

A similar technique is used in a more involved context in Section~\ref{sec:fg}.

\paragraph*{Back to point location.}
The analysis can then proceed following the approach in Section~\ref{sec:line}.
That is, using Fredman's trick, we locate the points $(a,a')$ in the arrangement 
of the set of strata of the surfaces~$\sigma_{t,t'}$, 
for $(a,a')\in \bigcup_\tau (A_\tau\times A_\tau)$ and for
$(t,t')\in \bigcup_{\tau'} (B_{\tau'}\times B_{\tau'})$, or locate the 
dual points $(t,t')$ in the arrangement of the strata of the curves $\delta_{a,a'}$, 
for the same pairs $(a,a')$, $(t,t')$, with the same performance cost as 
in Section~\ref{sec:line}. 

Concretely, we take the set $\Sigma$ of all strata of the surfaces $\sigma_{t,t'}$, 
of dimension at most $3$, and construct a suitable cutting of $\Sigma$. Since
some strata of $\Sigma$ are of dimension $0$, $1$, or $2$, one has to modify the
original construction (as presented in \cite{CEGS-91,Koltun-04}) to accommodate such surfaces 
too. This is an easy technical modification which we omit here.
A similar cutting is constructed in the plane at any dual step.

If $\kappa$ is a full-dimensional primal cell of the cutting, constructed
during some stage of the recursion, then each surface $\sigma_{t,t'}$
(which, as we recall, is the union of its strata) either intersects~$\kappa$ 
or is disjoint from $\kappa$. If $\kappa$ is of lower dimension, then it may also 
be fully contained in, or partially overlap, some strata of the surfaces $\sigma_{t,t'}$.


Similar to what is done in standard applications of cuttings, we say that a stratum $\phi$ 
\emph{crosses} a cell $\kappa$ of the cutting if $\phi\cap\kappa \ne \emptyset$
but $\kappa$ is not contained in $\phi$. In particular, if $\kappa$ is of lower 
dimension, strata that partially overlap $\kappa$ (but do not fully contain it) 
are also regarded as crossing $\kappa$.

The standard theory of cuttings (see, e.g., \cite{CF-90}) asserts that the 
number of strata that cross a cell $\kappa$, including those that partially 
overlap it (for lower-dimensional cells), is at most $|\Sigma|/r$, where $r$ 
is the parameter of the cutting.

With this observation, the recursive point-location mechanism
can proceed as before, and its output is, as before, a union of 
complete bipartite graphs, each of the form $\Sigma_\kappa\times P_\kappa$, 
where $\kappa$ denotes some primal or dual cell of the cutting constructed at some stage of 
the recursion, where in the primal, $P_\kappa$ is the set of primal points contained in $\kappa$, 
and $\Sigma_\kappa$ is a set of primal surfaces whose CAD has a stratum that fully contains $\kappa$
(there can be at most one such stratum per surface). In this definition we also consider
full-dimensional strata (which tile up the complement of the surface). Thus if $\kappa$
is fully contained in such a stratum, it means that the surface is disjoint from $\kappa$.
Symmetric definitions apply in the dual context.

Lemma~\ref{inv-sigma-delta} implies that for each 
bipartite graph $\Sigma_\kappa\times P_\kappa$ the following properties hold.
If the graph was constructed in a primal phase and $\kappa$ is full-dimensional
then, for each~$\sigma_{t,t'}\in\Sigma_\kappa$ (which misses $\kappa$), each of 
the sequences $\Xi(a,t)$, $\Xi(a',t')$ has a fixed combinatorial structure,
and each comparison between a root in $\Xi(a,t)$ and a root in $\Xi(a',t')$ has
a fixed outcome (which is a strict inequality), for all $(a,a')\in\kappa$. 
If $\kappa$ is of lower dimension and $\sigma_{t,t'}$ misses it, the same
argument applies, but it also holds when $\sigma_{t,t'}$ has a stratum that 
fully contains~$\kappa$.
Indeed, any change in any of these invariants
must occur when (in the primal) $(a,a')$ crosses~$\sigma_{t,t'}$, or, when 
$\kappa$ is not full-dimensional, when $(a,a')$ varies in some fixed stratum of
$\sigma_{t,t'}$ that partially overlaps $\kappa$, and reaches the 
boundary of this stratum. In either case, by construction, $\sigma_{t,t'}$ 
does not participate in the graph $\Sigma_\kappa\times P_\kappa$.
Symmetric properties hold when $\kappa$ is constructed in the dual phase. 


Note the order of quantifiers in our description of $\Sigma_\kappa \times P_\kappa$: We are not claiming that these properties hold
uniformly for every surface in $\Sigma_\kappa$ and every point in $P_\kappa$, but
rather that they hold for each surface uniformly over the points. Different
surfaces may have different outcomes, but these outcomes are fixed for each surface,
over the entire cell $\kappa$.

We still need to determine these outcomes. We use the following technique.
Suppose we are in the primal, and consider a cell $\kappa$ of the partition
at some recursive step. As argued, for each $\sigma_{t,t'}\in\Sigma_\kappa$, the outcomes
of all the comparisons between the roots of $\Xi(a,t)$ and those of $\Xi(a',t')$,
as well as the combinatorial structure of each of these two sequences,
are independent of the point~$(a,a')\in\kappa$. We therefore pick an
arbitrary representative point~$(a_\kappa,a'_\kappa)\in\kappa$, not necessarily from~$A\times A$,
and compute, explicitly and by brute force, the required outcomes (for the pairs
$(a_\kappa,t)$ and $(a'_\kappa,t')$), which therefore give us all the outcomes 
for all the pairs arising from $\{\sigma_{t,t'}\}\times P_\kappa$. By repeating
this for all $\sigma_{t,t'}\in\Sigma_\kappa$, we resolve all the comparisons
encoded in $\Sigma_\kappa\times P_\kappa$. A fully symmetric procedure applies
in the dual context. Repeating this step for all cells $\kappa$, both primal and dual,
we resolve all comparisons needed to sort the $s$-roots. The total time for this step
is proportional to the total size of the vertex sets of the complete bipartite graphs,
and is thus within the time bound of the algorithm in Section~\ref{sec:line}.

Once we have all those sorted sequences, searching with the values $s\in S$ 
is done exactly as in Section~\ref{sec:line}.

We thus obtain: 


\begin{theorem}
  \label{thm:single_pol}
  Let $A$, $B$, $C$ be $n$-point sets in the plane, 
  where $B$ and $C$ are each contained in some respective
  constant-degree algebraic curve $\gamma_B$, $\gamma_C$, with parametric representations
  $\gamma_B(t) = (x_B(t),y_B(t))$ and $\gamma_C(s) = (x_C(s),y_C(s))$,
  for $t,s\in \reals$, where $x_B(t)$, $y_B(t)$, $x_C(s)$, $y_C(s)$ are constant-degree 
  polynomials (or rational functions). Let $T$ and $S$ be the 
  respective sets of parameter values that represent the points of $B$ and of $C$.
  Let $F$ be a constant-degree 6-variate real polynomial.
  Then one can test, in the algebraic decision-tree model, whether there exists a triple
  $(a,t,s)\in A\times T\times S$, such that $H(a,t,s) \coloneqq F(a,\gamma_B(t),\gamma_C(s)) = 0$, 
  using only $O\left(n^{28/15 + \eps}\right)$ polynomial sign tests (in expectation), for any $\eps > 0$,
  where the constant of proportionality depends on $\eps$ and on the degree of~$H$.
\end{theorem}

\noindent\textbf{Remark.} 
As already noted, Theorem~\ref{thm:single_pol} can be extended to the case where
the parameterization functions $x_B(t)$, $y_B(t)$, $x_C(s)$, $y_C(s)$ are arbitrary
constant-degree algebraic functions, but we do not treat this extension in this work,
so as to avoid the extra technical complications that this extension involves.

As an immediate special case, we obtain:
\begin{theorem}[General collinearity testing in $2{\times}1{\times}1$ dimensions]
  \label{thm:coll}
  Let $A$, $B$, $C$ be $n$-point sets in the plane as in Theorem~\ref{thm:single_pol}. 
  Then one can test whether $A\times B\times C$ contains a collinear triple, in
  the algebraic decision-tree model, using only
  $O\left(n^{28/15+\eps}\right)$ polynomial sign tests (in expectation), for any $\eps>0$,
  where the constant of proportionality depends on $\eps$ and on the 
  degrees of $\gamma_B$, $\gamma_C$.
\end{theorem}

As a more general corollary of Theorem~\ref{thm:single_pol}, we obtain: 

\begin{corollary}[Unit area triangles and related problems]
  \label{cor:unit_area}
  Let $A$, $B$, $C$ be $n$-point sets in the plane as in Theorem~\ref{thm:single_pol}. 
  Then one can test whether $A\times B\times C$ contains a triple that form a unit-area triangle,
  or a unit-perimeter triangle, or a triangle with a prescribed circumradius or inradius,
  in the algebraic decision-tree model, using only
  $O\left(n^{28/15+\eps}\right)$ polynomial sign tests (in expectation), for any $\eps>0$,
  where the constant of proportionality depends on $\eps$ and on the 
  degrees of $\gamma_B$, $\gamma_C$.
\end{corollary}

\subsection{An Alternative Approach}
\label{sec:alternative}

Recall that we are given $n$-point sets $A$, $B$, and $C$ in the plane, with $A$ unrestricted, 
while each of $B$ and $C$ is contained in its respective constant-degree algebraic curve  
(which we regard, for simplicity, as an injective image of the real line, via our
parameterizations). In what follows 
we present a different approach to handling the $2{\times}1{\times}1$-dimensional
collinearity testing, in which we flip the roles of $A$ and $C$.
That is, we preprocess $B\times C$ and then search with the points of~$A$.
This approach yields a subquadratic bound inferior
to the one presented in Theorem~\ref{thm:coll}, %
but its methodology is interesting in its own right, 
and we hope that it will find applications in other problems of this kind.
In particular, we use this approach in a follow-up work in progress (see the introduction for a brief description thereof).
It is also perhaps closer to the approach of Barba \etal~\cite{BCILOS-17} and of Chan~\cite{Chan-18}.

The preprocessing of $B\times C$ is easy, as these sets are one-dimensional.
We choose some parameter $g$ (this time, it is the same %
for~$B$ and~$C$), whose value will be fixed later, and partition $B$ (resp., $C$)
by subdividing $\gamma_B$ (resp., $\gamma_C$) into $n/g$ arcs, each 
containing $g$~points of~$B$ (resp., $C$). We denote the resulting blocks
of $B$ (resp., $C$) as~$B_1,\ldots,B_{n/g}$ (resp.,~$C_1,\ldots,C_{n/g}$).

Let $H$ be the $4$-variate (i.e., $2{\times}1{\times}1$-variate) constant-degree
polynomial representing the collinearity testing function in~(\ref{eq:collinear}).
Our goal is to test whether it vanishes at some point of 
$A\times T\times S$, where
$T$ and $S$ are the sets of real numbers that represent $B$ and $C$, as above.
For each $t\in T$, $s\in S$, we define the curve
\[
\delta_{t,s} \coloneqq \{a\in \reals^2 \mid H(a,t,s) = 0\} ,
\]
which is in fact a line (refer once again to~(\ref{eq:collinear})). 
Let $\Gamma$ be the collection of these lines. 
In principle, we want to construct the arrangement~$\A(\Gamma)$ of these lines, 
and locate the points of~$A$ in the arrangement, aiming to detect points that lie on a line 
of $\Gamma$. This is too expensive, even in the algebraic decision-tree model,
so we only aim, following, as above, the usual spirit of Fredman's trick,
to construct the subarrangements~$\A(\Gamma_{i,j})$, over all
pairs $(i,j) \in [1,n/g]^2$, where 
\[
\Gamma_{i,j} \coloneqq \{\delta_{t,s} \mid t\in T_i,\; s\in S_j \} ,
\]
and $T_i$ (resp., $S_j$) is the subset of numbers representing the points of $B$ (resp., $C$)
in the block $B_i$ (resp., $C_j$).
Even this more modest goal is still too expensive in the uniform model, but
we can make it efficient (i.e., significantly subquadratic) in the algebraic 
decision-tree model. To do so, we play the following variant of Fredman's trick. 

Ignoring for the moment efficiency of the procedure, performing point location, 
with some point $a=(a_1,a_2)\in\reals^2$, in $\A(\Gamma_{i,j})$ is 
accomplished by performing two types of operations:
(i)~Determine whether $a_1$ lies (lexicographically) to the left or to the right of some
vertex of~$\A(\Gamma_{i,j})$.
(ii)~Test whether $a$ lies above or below some line 
of $\Gamma_{i,j}$.
We focus on operations of type (i) and will later argue that we can implement the procedure 
so that operations of type~(ii) are fairly easy to perform. To speed up the procedure, we perform 
the point location, and the preceding preprocessing, in the ``combined'' arrangement $\A(\cup_{i,j} \Gamma_{i,j})$.

In order to be able to perform operations of type (i), we simply sort
the vertices of~$\A(\cup_{i,j}\Gamma_{i,j})$ lexicographically in their $x$-order (that is, 
first according to their $x$-coordinates, and then according to their $y$-coordinates if two vertices lie on a vertical line),
and then locate each query point $a\in A$ amid them by binary search (with the $x$-coordinate 
of $a$, and then with its $y$-coordinate) through this sequence.
To sort the vertices, we need to (lexicographically) compare the $x$-coordinates of pairs 
of vertices. In general, each such pair $(u,v)$ is determined by four lines, 
with two lines, 
$\delta_{t_1,s_1}$, $\delta_{t_2,s_2}$ intersecting at $u$, and the other two lines, 
$\delta_{t_3,s_3}$, $\delta_{t_4,s_4}$ intersecting at $v$.
To simplify the presentation, assume that the test of whether $u$ lies (lexicographically) to the 
left of $v$ amounts to testing the sign of some fixed-degree polynomial
$\Phi(t_1,t_2,t_3,t_4;s_1,s_2,s_3,s_4)$.
(In full generality this may involve deciding whether some Boolean predicate involving 
polynomial equalities and inequalities in these variables is satisfied.)

%

We can now play Fredman's trick combined with duality. In the primal space
we regard $(t_1,t_2,t_3,t_4)$ as a point in $\reals^4$, and associate with 
each quadruple $(s_1,s_2,s_3,s_4)$ the three-dimensional surface 
\[
\sigma_{s_1,s_2,s_3,s_4} \coloneqq \{(t_1,t_2,t_3,t_4) \mid
\Phi(t_1,t_2,t_3,t_4;s_1,s_2,s_3,s_4) = 0 \} .
\]
Let $\Sigma$ denote the collection of these surfaces, gathered from
all quadruples $(s_1,s_2,s_3,s_4) \in \cup_j S_j^4$, and let $P$ 
denote the set of all points $(t_1,t_2,t_3,t_4) \in \cup_i T_i^4$.
By locating the points of $P$ in the arrangement $\A(\Sigma)$, we obtain
the answers to all the relevant sign tests involving $\Phi$.
The symmetry of the setup allows us to dualize the problem,
treating the quadruples $(s_1,s_2,s_3,s_4)$ as points in $\reals^4$, and 
associating with the quadruples $(t_1,t_2,t_3,t_4)$ three-dimensional surfaces,
defined in a fully symmetric manner.  We have 
\[
\abs{P},\; \abs{\Sigma} = O((n/g)\cdot g^4) = O(ng^3) .
\]
Applying suitable partitionings both in the primal and dual space,
as we did in Lemma~\ref{lem:107} (see also Lemma~\ref{lem:85}),
we can perform the point locations in time $O((ng^3)^{8/5+\eps})$, for any~$\eps>0$.
So far, we have managed to (lexicographically) sort the vertices of each of the
arrangements $\A(\Gamma_{i,j})$ by their $x$-coordinates. 
In the next step of the point-location mechanism, we partition $\A(\Gamma_{i,j})$ into vertical slabs by drawing
a vertical line through each vertex, and note that within each slab the lines
of $\Gamma_{i,j}$ are pairwise disjoint, so they have a fixed
vertical order, allowing us to locate any point $a$ in the slab amid them
by a simple binary search. Fortunately, this part of the preprocessing
hardly costs anything in the algebraic decision-tree model. 
Indeed, we only need to sort the lines 
in each $\Gamma_{i,j}$ by their vertical order at $x=-\infty$, namely
in reverse order of their slopes, which we can do in
$O(ng^3\log(ng))$ steps. We then simply sweep the arrangement
from left to right, and update the sorted sequence of the lines 
in each new slab that we reach, which is obtained from the sequence of the previous 
slab by swapping the order of the two (or more) lines 
that meet at the vertex that we sweep through.
This sweep and updates cost nothing in our model.

As before, we next want to search in the arrangements $\A(\Gamma_{i,j})$
with the points of $A$, and we argue that each point needs to be located
in only a small number of arrangements. For this, for each $a\in A$,
we look at the curve $\delta_a = \{(t,s) \mid H(a,t,s) = 0 \}$, and note
the easily verified property that it crosses only $O(n/g)$ blocks $B_i\times C_j$ (that is, $T_i \times S_j$). 
Hence, the overall number of searches that we have to perform is 
$O(n\cdot(n/g)) = O(n^2/g)$. Each search takes $O(\log g)$ steps. 
Hence the overall cost of the procedure (in the algebraic decision-tree model) is
\[
O\left((ng^3)^{8/5+\eps} + \frac{n^2\log g}{g} \right) .
\]
Once again, roughly balancing the two terms, we choose $g = n^{2/29}$ and obtain
a total cost of~$O(n^{56/29+\eps})$, for any $\eps>0$.
This, albeit subquadratic, is a weaker bound than the one in Theorem~\ref{thm:coll}.

\section{Testing for the Vanishing of a Pair of Polynomials, and \\
Collinearity Testing in the Complex Plane}
\label{sec:fg}

In this section we study problems of type~(b), 
where we are given
three arbitrary sets $A$,~$B$, and~$C$ of $n$~points in the plane, and we seek a triple  
$(a,b,c)\in A\times B\times C$ that satisfies two polynomial equations of the form
$F(a,b,c) = 0$, $G(a,b,c) = 0$, where $F$ and $G$ are two
non-zero constant-degree 6-variate real polynomials without a common factor. 

As in Section~\ref{sec:col}, handling the general case, of arbitrary polynomials 
$F$ and $G$, raises technical issues that involve various degenerate scenarios.
Here, though, treating these degeneracies in full generality either leads to subproblems 
that are rather complicated to handle, or that we do not know how to solve efficiently.
As a consequence, we will be making several assumptions on~$F$ and~$G$, and some of them
will have to be checked by ad-hoc techniques for any concrete instance of the problem. 
Complex collinearity testing, as posed in the introduction, is one of these instances
(our flagship instance as a matter of fact), and we will note, after each assumption 
that we make, how it specializes to collinearity testing, and, if applicable,
what simpler assumptions are needed to make it work. (Most of these specializations and 
alternative assumptions turn out to be relatively simple for complex collinearity testing.)
This is somewhat different from the treatment in Section~\ref{sec:col}, where, 
for most of the assumptions made, we had a general procedure that either verifies them
or solves the problem, positively or negatively, in an efficient straightforward manner
when the assumptions do not hold. Such procedures are also available for some, 
but not for all, the assumptions made in this section.

\paragraph*{$F$ and $G$ have good fibers.}
Similar to Section~\ref{sec:col}, one of our main assumptions is that the polynomials $F$ and $G$
have \emph{good fibers}, in the sense that, for every $c_0\in C$, the surface 
$\pi_{c_0} \coloneqq \{(a,b)\in \reals^4 \mid F(a,b,c_0) = G(a,b,c_0) = 0 \}$ 
is two-dimensional
and has good fibers, meaning (see Section~\ref{sec:prelim}) that, for each $c_0\in C$,
the system $F(a,b_0,c_0) = G(a,b_0,c_0) = 0$ has only $O(1)$ solutions in $a$ for each $b_0\in B$, 
possibly excluding $O(1)$ exceptional values of $b_0$, and the system 
$F(a_0,b,c_0) = G(a_0,b,c_0) = 0$ has only $O(1)$ solutions in $b$ for each $a_0\in A$, 
possibly excluding $O(1)$ exceptional values of~$a_0$.

In Section~\ref{sec:good} we had a procedure that tests whether the single polynomial $F$ has good fibers,
and solves the problem in a simple manner when $F$ does not have this property. Here, in contrast,
if $F$ and $G$ do not have the good fibers property, we may find ourselves in situations where we do not
know how to find a solution efficiently. We will comment on this issue later in the section.

Note that, for simplicity, and unlike the setup in Section~\ref{sec:col}, 
we do not allow exceptional values of~$c_0$, at which the above
properties do not hold. If we have an instance in which the number of such exceptional values
is finite (and therefore constant), and we have an efficient procedure for enumerating them, we could then solve the problem for each of these values separately. In the worst case,
such a subproblem could become that of having a single polynomial that we want to test whether
it vanishes at some point of $A\times B$ (this happens when one polynomial is identically zero or the 
zero sets of the two polynomials have a common component), but even this problem is easier to solve than 
the main one we face in this section. However, when there are infinitely many exceptional values 
$c_0$, handling them may be problematic---see the promised discussion later in the section.

Observe that the good fibers assumption implies that 
$F$ and $G$ do not have a nontrivial common factor. 
Indeed, the existence of such a factor $\phi$ would imply that, for each $c_0\in C$, 
the surface $\pi_{c_0}$ contained the surface $\{(a,b)\in \reals^4 \mid \phi(a,b,c_0) = 0 \}$, 
which would be three-dimensional, contrary to the assumption.

(For complex collinearity testing, we first assume that the sets $A$, $B$, $C$ are pairwise disjoint,
for otherwise any pair of coinciding points form a collinear triple with any point from the third set; this can be efficiently tested for.
We first need to rule out exceptional values of $c_0\in C$ for which the surface 
$\{(a,b)\in \gamma_A\times\gamma_B \mid ab \text{ is collinear with } c_0\}$ has real dimension at least three.
This can happen only when $\gamma_A$ and $\gamma_B$ are a pair of coinciding lines that pass 
through $c_0$. (In any other case, we can locally parameterize the surface by a point $z$ on
one of these curves, for which the line through $c_0$ and $z$ intersects the other curve in only 
$O(1)$ points (unless the other curve is a line that passes through $c_0$ and $z$,
a situation that we will shortly filter out), showing that the surface has real dimension two.)
We therefore assume that this does not occur. As a matter of fact, we assume the
stronger property that if $\gamma_A$ (resp., $\gamma_B$) is a line then it is different from 
$\gamma_B$ and $\gamma_C$ (resp., from~$\gamma_A$ and~$\gamma_C$). Under this assumption, the above
surface has real dimension two for every $c_0\in C$. The non-exceptionality of $b_0$, for a fixed $c_0$, 
means that the complex line through $c_0$ and $b_0$ contains only $O(1)$~points of $A$ 
(symmetrically, the non-exceptionality of $a_0$ means that the complex line through 
$c_0$ and $a_0$ contains only $O(1)$~points of $B$). 
This holds, by our assumption, unless $\gamma_A$ is a line that contains both $c_0$ and $b_0$
(or $\gamma_B$ is a line that contains both $c_0$ and $a_0$). Testing for the existence of
any such pair of points $(c_0,b_0)\in C\times B$ or $(c_0,a_0)\in C\times A$ is easy to do
in linear time, and the existence of any such pair implies that there are (many) collinear triples.
It is therefore safe to assume that, if $\gamma_A$ (resp., $\gamma_B$) is a line, it does not 
meet both $B$ and $C$ (resp., $A$ and $C$). With all these assumptions, we have the desired
good-fibers property. We also note that the case (which we want to rule out) where a pair of 
the curves are coinciding lines is easy to handle efficiently, since all collinearities must occur on 
that line, because our sets are assumed to be pairwise disjoint, 
and they occur if and only if the third set intersects that line.)

\paragraph*{A high-level description of the algorithm.}

We fix a (single) parameter $g\ll n$ (whose value will be set later), 
and apply Theorem~\ref{thm:hier_partition_CP}(i) to the sets $A$, $B$,
with $r = g$, to construct, implicitly and in expected $O(n\log n)$ time, 
a hierarchical polynomial partitioning 
for~$A \times B$, obtained from hierarchical planar polynomial partitionings for $A$ and for $B$,
so that each bottom-level cell is associated with a Cartesian product $A'\times B'$ of
subsets $A'\subset A$, $B'\subset B$, satisfying $|A'|$, $|B'| \le g$ 
(so the cell is associated with at most $g^2$ points of $A \times B$), 
and the overall number of cells is~$O((n/g)^{2 + \eps})$, for any prescribed $\eps > 0$.

Let $\tau$ (resp., $\tau'$) be a bottom-level cell in the hierarchical partition of~$A$
(resp., of~$B$). Let $A_\tau$ be the set of points of $A$ associated with $\tau$, and
$B_{\tau'}$ be the set of points of $B$ associated with $\tau'$.
The high-level idea of the algorithm is as follows. For most pairs 
$(a,b)\in A_\tau \times B_{\tau'}$, the system $F(a,b,c) = G(a,b,c) = 0$
has a constant number of distinct roots in $c$ (that is, in the two 
coordinates $c_1$, $c_2$ of $c$), which we can enumerate in sorted lexicographic order
as $\rho_1(a,b) < \ldots < \rho_k(a,b)$, where $k = k_{a,b}$ also depends on $a$ and $b$.
Similar to Section~\ref{sec:col}, we denote this sequence as~$\Xi(a,b)$.

The locus $\Omega^\infty$ of those $(a,b)$ for which the system has infinitely many solutions in $c$
is a lower-dimensional surface, which is obtained by eliminating one of the variables $c_1$,~$c_2$.
Specifically, for at least one of these variables, say $c_2$, its elimination must 
yield the identically zero polynomial in $c_1$, so all its coefficients (polynomials 
in $a,b$) must vanish, and their common zero set is the desired locus (see~\cite{CLO}).
We collect all these discrete roots, over all pairs~$(a,b)\in A_\tau \times B_{\tau'}$
for which the roots are indeed discrete, and sort them into a single sequence $\Lambda_{\tau,\tau'}$. 

The exceptional pairs $(a,b)\in\Omega^\infty\cap (A\times B)$ are handled as follows. We assume 
that, for each $c\in C$, the two-dimensional surface $\pi_c$ meets $\Omega^\infty$ (which is of 
dimension at most~$3$) in a one-dimensional curve, denoted as $\delta_c$.\footnote{%
  This is one of our assumptions listed towards the end of this section, which we need to make in order to have our algorithm work properly---see below.} 
By Theorem~\ref{thm:hier_partition_CP} (see also Appendix~\ref{app:hierarchical}),
$\delta_c$ crosses only $O((n/g)^{1/2+\eps})$ cells $\tau\times\tau'$, and these cells can be 
computed in time $O((n/g)^{1/2+\eps})$;
see Corollaries~\ref{app:cor:hier_partition_1dim} and~\ref{app:cor:crossing_time}.
When we search with $c$, we take each of these cells and check all of its $g^2$ 
elements $(a,b)$ for the vanishing of $F$ and $G$ at $(a,b,c)$. The overall cost,
for all points of $c \in C$ and for all the cells that contain exceptional pairs $(a,b)$ 
and are crossed by $\pi_c$, 
is therefore $O(n(n/g)^{1/2+\eps}g^2) = O(n^{3/2+\eps}g^{3/2-\eps})$, 
which is subsumed in the bound on the complexity of the other steps of the algorithm.

(For complex collinearity testing, $(a,b)\in \Omega^\infty$, i.e., $(a,b)$ is a pair for which
the system $F(a,b,c) = G(a,b,c) = 0$ has a continuum of roots in $c$,
if and only if the line $ab$ coincides with $\gamma_C$. By our previous assumptions,
this scenario does not occur.)

Consider next the case where $(a,b)\notin\Omega^\infty$, so $\Xi(a,b)$ is finite and
forms part of the corresponding sequence $\Lambda_{\tau,\tau'}$.
We then search with each $c\in C$ through those sorted sequences
$\Lambda_{\tau,\tau'}$ that might contain $c$. We show, under the good-fibers 
assumption, that each~$c\in C$ has to be searched for in only $O((n/g)^{1+\eps})$
sequences. We succeed if we find~$c\in C$ for which one of these searches identifies 
$c$ as a member of the corresponding sequence.

As already noted in the previous section, sorting the union of the sets $\Xi(a,b)$ 
explicitly is too expensive, as their overall size is $O(n^2)$. 
We overcome this issue by considering the problem in the 
algebraic decision-tree model, and by using an algebraic variant of Fredman's trick
in the same spirit as in Section~\ref{sec:col}.

In sorting $\Lambda_{\tau,\tau'}$, we have to perform various (lexicographic) comparisons 
of pairs of roots $\rho_i(a,b)$ and $\rho_j(a',b')$, for $a, a'\in A_\tau$, $b, b'\in B_{\tau'}$,
and indices~$i$,~$j$. The lexicographical order can change when
the $x$-coordinates of $\rho_i(a,b)$ and $\rho_j(a',b')$ become equal 
(and then switch their order), for $a, a'\in A_\tau$, $b, b'\in B_{\tau'}$,
or when the two roots actually coincide. The former situation imposes one equality
and thus is expected to yield a three-dimensional surface, whereas the latter
situation imposes two equalities, and thus is expected to yield a two-dimensional surface.

We thus transform these comparisons to the following setup.
We consider $B_{\tau'}\times B_{\tau'}$ as a set of~$g^2$~points in~$\reals^4$, 
and associate with each pair $(a,a') \in A_{\tau}\times A_{\tau}$
the three-dimensional surfaces
\begin{equation}
  \begin{aligned}
    \sigma^x_{a,a'} \coloneqq \{ & (b,b') \in\reals^4 \mid \exists x,y,y'\in\reals :
    y\neq y' \text{ and } 
    \\ 
    & F(a,b,(x,y)) = G(a,b,(x,y)) = F(a',b',(x,y')) = G(a',b',(x,y')) = 0 \} , \\
    \sigma^y_{a,a'} \coloneqq \{ & (b,b') \in\reals^4 \mid \exists x,x',y\in\reals :
    x\neq x' \text{ and } 
    \\ 
    & F(a,b,(x,y)) = G(a,b,(x,y)) = F(a',b',(x',y)) = G(a',b',(x',y)) = 0 \} . 
  \end{aligned}\label{eq:all-sigma}
\end{equation}

Extending the technique in Section~\ref{sec:col}, a three-step elimination of $x$, $y$, and $y'$,
or of $x$, $x'$, and~$y$, respectively, yields quantifier-free semi-algebraic expressions for these surfaces, which are
of more involved form, but still of constant complexity, and which show that, typically, these
are three-dimensional surfaces.
(This step costs nothing in our model.)

As in Section~\ref{sec:col}, a technical issue that arises here is that these surfaces 
are defined only in terms of $a$ and $a'$, and so they are ``disentangled'' from the points
$b$ and $b'$. As $b$~and $b'$~vary, the \emph{combinatorial structure} of
$\Xi(a,b)$ or of $\Xi(a',b')$ may change. As this notion is somewhat more involved in
the present setup, we spell it out in more detail.

Before doing so, we decompose each curve $F(a,b,\cdot)=0$, and each curve 
$G(a,b,\cdot)=0$, into \emph{strata}, using the CAD construction (recall 
the setup discussed in Section~\ref{sec:1poly}). \footnote{%
  We could also construct the full six-dimensional CAD of $F$ or of $G$,
  and extract the information from there, but this does not affect the
  present analysis.}

We then have:
\begin{definition}[Fixed combinatorial structure of $\Xi(a,b)$]
We say that the sequence $\Xi(a,b)$ of $c$-roots 
has a \emph{fixed combinatorial structure}, as $a$ and $b$ vary,
if the number $k_{a,b}$ of distinct roots (the size of~$\Xi(a,b)$) is fixed, 
and each of the roots lies on a fixed stratum of the curve $F(a,b,\cdot)=0$ 
and on a fixed stratum of the curve $G(a,b,\cdot)=0$.
In addition, no pair of roots that are not 
co-vertical become co-vertical, every co-vertical pair of roots remains 
co-vertical, and the $y$-order of the roots in such a pair does not change.
Roots that have multiplicity greater than one retain their multiplicity.
\end{definition}

We note that, for a pair $a,b$, the combinatorial structure of $\Xi(a,b)$ can 
change only at values~$(a,b)$, for which the system $F(a,b,c)=G(a,b,c)=0$ has 
either (i)~a double root in $c$, (ii)~a root that lies at a singular point 
of one of the curves, (iii)~infinitely many roots, or (iv) two co-vertical roots. 
The first type of criticality can happen when the 
curves $F(a,b,c) = 0$ and $G(a,b,c) = 0$ (in the $c$-plane, with~$a$,~$b$~fixed) 
are tangent to each other. The second type of criticality happens when one 
curve passes through a singular point of the other (this also covers cases where 
two roots that lie on different strata of one or two of the curves coincide, 
or a double root splits in the reverse manner). The third type of criticality,
namely where $(a,b) \in \Omega^\infty$, happens when the curves coincide, or come to 
have a common component, or one curve degenerates to the entire plane. 
Finally, the fourth type of criticality happens when a pair of roots become 
co-vertical, or stop being co-vertical, or when two co-vertical roots coincide 
or split (these latter situations are also criticalities of the first type). 

We note that in some of these criticalities (e.g., of the first type), the
\emph{Jacobian} 
$$
\begin{vmatrix}
\frac{\bd F}{\bd c_1} & \frac{\bd F}{\bd c_2} \\
& \\
\frac{\bd G}{\bd c_1} & \frac{\bd G}{\bd c_2} 
\end{vmatrix}
$$
of the system $F=G=0$ vanishes at $(a,b,c)$, which is an easy consequence of 
the implicit function theorem (see \cite{BCR-98}). We therefore also introduce
criticalities of type (v), of points $(a,b)$ for which the Jacobian vanishes
at one of the roots. This addition will also be useful for establishing the
continuity of the roots as $(a,b)$ vary; see below for details.


We take the locus of those $(a,b)$ for which any of the criticalities of 
types (i)--(v) occurs, which includes $\Omega^\infty$, and denote the 
resulting set as $\Omega$. It is a semi-algebraic set 
in the four-dimensional $ab$-space, of dimension at most $3$. The
smaller dimensionality follows from the fact that in these critical points
we have additional constraints (such as the vanishing of the Jacobian), which reduce the dimension.
We assume that $\Omega$ has good fibers, in the sense that for each $a\in A$,
with the possible exception of $O(1)$ values, the set 
$\Omega_{(a,\cdot)} \coloneqq \{b\in\reals^2 \mid (a,b)\in\Omega\}$ is 
a one-dimensional curve or a discrete set, and for each $b\in B$, 
with the possible exception of $O(1)$ values, the set 
$\Omega_{(\cdot,b)} \coloneqq \{a\in\reals^2 \mid (a,b)\in\Omega\}$ is 
a one-dimensional curve or a discrete set.
Since $\Omega$ is of constant complexity, we can test whether this 
property holds (as a global statement, over all $a$, $b$), in the 
algebraic model that we follow, in constant time. (It actually costs nothing
in the algebraic decision-tree model.)

We remark that, if the assumption that $\Omega$ has good fibers fails,
then, up to reversal of the roles 
of~$A$ and~$B$, there is a curve $\omega$ in the $a$-plane such that
$\omega \times \reals^2 \subseteq \Omega$. One of several degeneracies
that may arise in such a case is that one of our polynomials, say~$G$,
becomes identically zero in~$c$ over~$\omega \times \reals^2$
(this arises, e.g., at criticalities of type (iii)), and then the problem reduces to 
that of a single vanishing polynomial in $1\times 2\times 2$ dimensions
(when the majority of the points in $A$ lie on $\omega$),
a problem that we do not know how to solve in subquadratic time.
Hence this assumption is indeed needed.

(For complex collinearity testing, $(a,b)\in \Omega$ means that the line
through $a$ and $b$ is either tangent to $\gamma_C$ or passes through a 
singular point of that curve, or overlaps $\gamma_C$ in case $\gamma_C$ is a line.
The assumptions already made imply that the latter scenario does not happen,
and the former scenarios can happen, for each $a\in A$, for only $O(1)$ points $b\in B$,
and it can happen, for each $b\in B$, for only $O(1)$ points $a\in A$,
showing that the property holds in this problem.)

Next, returning to the surfaces $\sigma^x_{a,a'}$ and $\sigma^y_{a,a'}$ 
(see~\eqref{eq:all-sigma}), we unite them, for each pair $(a,a')\in A\times A$,
and add to them the surfaces $\Omega_{(a,\cdot)} \times \reals^2$ and
$\reals^2 \times \Omega_{(a',\cdot)}$. We denote the resulting surface
as $\sigma_{a,a'}$. As in Section~\ref{sec:col}, we construct the CAD of
$\sigma_{a,a'}$, and call the cells of the CAD,
of any dimension, the \emph{strata} of $\sigma_{a,a'}$.
We denote by $\Sigma$ the set of the strata of all these surfaces.

As in Section~\ref{sec:col}, we will also be working in a dual setup,
in which the roles of $A$ and $B$ are switched. In the dual context
we define, in complete analogy, for each pair $(b,b')\in B\times B$, the dual surfaces
\begin{equation}
  \begin{aligned}
    (\sigma^*)^x_{b,b'} \coloneqq & \{ (a,a') \in\reals^4 \mid \exists x,y,y'\in\reals : 
    y\neq y' \text{ and } \\
    & F(a,b,(x,y)) = G(a,b,(x,y)) = F(a',b',(x,y')) = G(a',b',(x,y')) = 0 \} \\
    (\sigma^*)^y_{b,b'} \coloneqq & \{ (a,a') \in\reals^4 \mid \exists x,x',y\in\reals : 
    x\neq x' \text{ and } \\
    & F(a,b,(x,y)) = G(a,b,(x,y)) = F(a',b',(x',y)) = G(a',b',(x',y)) = 0 \} , 
  \end{aligned}
  \label{eq:all-sigma-dual}
\end{equation}
and form the surfaces $\Omega_{(\cdot,b)} \times \reals^2$ and
$\reals^2 \times \Omega_{(\cdot,b')}$. We unite all these surfaces, 
for each fixed pair $(b,b')$, and denote the resulting surface as $\sigma^*_{b,b'}$.
Here too we construct the CAD of this surface, take the cells of the CAD to be
the (dual) \emph{strata} of $\sigma^*_{b,b'}$, and denote by $\Sigma^*$ the
collection of these dual strata, over all surfaces $\sigma^*_{b,b'}$.

As in Section~\ref{sec:col}, the following lemma shows that the strata of the
surfaces $\sigma_{a,a'}$, or, equivalently, the strata of the dual surfaces 
$\sigma^*_{b,b'}$, capture all the possible criticalities that quadruples
$(a,a',b,b')$ may exhibit. More precisely we have:

\begin{lemma} \label{inv-sigma-sigmastar}
Let $\Delta$ be a stratum, that is, a cell, of any dimension (including four-dimensional
strata which tile up the complement of $\sigma_{a,a'}$), 
of the CAD formed by the surface $\sigma_{a,a'}$. 
Then, as $(b,b')$ varies continuously within $\Delta$, the following properties hold. \\
(i) Each of the sequences $\Xi(a,b)$ and $\Xi(a',b')$ has a fixed combinatorial structure. \\
(ii) The sorted merged sequence of $\Xi(a,b)$ and $\Xi(a',b')$ (which, as in 
Section~\ref{sec:1poly}, we denote as $\Lambda(a,b;a',b')$) also has a fixed 
combinatorial structure, and any coincidence of roots of one sequence with roots of the 
other sequence remains invariant. In particular, for each $i=1,\ldots,k_{a,b}$, 
$j=1,\ldots,k_{a',b'}$, the relative lexicographic order of the $i$\emph{th} root of~$\Xi(a,b)$ and 
the $j$\emph{th} root of~$\Xi(a',b')$ (including possible equality between them) remains invariant.
(iii) Each of the roots of either sequence $\Xi(a,b)$, $\Xi(a',b')$, or a common 
root of both in $\Lambda(a,b;a',b')$, varies continuously with $(b,b')\in\Delta$.

\medskip
\noindent
Symmetrically, if $(a,a')$ varies continuously within a dual stratum $\Delta$, of any dimension, 
of the decomposition formed by the surface $\sigma^*_{b,b'}$, then
properties~(i)--(iii) hold for~$(a,b)$ and~$(a',b')$, with obvious modifications.
\end{lemma} 
\begin{proof}
We prove only the first part of the lemma; the second part is handled in a fully symmetric manner.
The proof follows the arguments in Section~\ref{sec:1poly}. That is,
we argue that, as long as $(b,b')$ varies within $\Delta$, none of the 
properties (i)--(iii) is violated. Indeed, for Property (i) to be violated, 
$(b,b')$ must reach another (at most three-dimensional) stratum within the surfaces $\Omega_{(a,\cdot)}$
or $\Omega_{(a',\cdot)}$, so in particular it reaches the boundary of $\Delta$.

For Property (ii), a coincidence between a pair of roots, one from each sequence 
$\Xi(a,b)$, $\Xi(a',b')$, can disappear or newly appear only when $(b,b')$ reaches another
stratum of $\sigma_{a,a'}$, as easily follows from the definition of this surface.

Concerning Property (iii), each root $c$, say of $\Xi(a,b)$, 
varies continuously as long as the Jacobian 
$$
\begin{vmatrix}
\frac{\bd F}{\bd c_1} & \frac{\bd F}{\bd c_2} \\
& \\
\frac{\bd G}{\bd c_1} & \frac{\bd G}{\bd c_2} 
\end{vmatrix}
$$
does not vanish at $(a,b,c)$, as implied by the implicit function theorem (see \cite{BCR-98}),
and, by construction, any point at which this is violated lies in $\Omega_{(a,\cdot)}\times\reals^2$.

The argument for the dual surfaces is fully symmetric.
\end{proof}

%

\paragraph*{Back to point location.}
The analysis can then proceed following the approach in Section~\ref{sec:1poly},
except that now the point location is done in four dimensions, both in the primal
and in the dual (see below for more details).
That is, using Fredman's trick, we locate the points $(b,b')$ in the arrangement 
of the set of strata of all the surfaces~$\sigma_{a,a'}$,
for $(a,a')\in \bigcup_\tau (A_\tau\times A_\tau)$ and for
$(b,b')\in \bigcup_{\tau'} (B_{\tau'}\times B_{\tau'})$, or locate the 
dual points $(a,a')$ in the arrangement of the strata of all the dual surfaces $\sigma^*_{b,b'}$, 
for the same pairs $(a,a')$, $(b,b')$.

Concretely, we take the set $\Sigma$ of all strata of all the surfaces $\sigma_{a,a'}$, 
of dimension at most $3$, and construct a suitable cutting of $\Sigma$, modified as in
Section~\ref{sec:1poly}, so that it also takes care of strata of dimension $0$, $1$, or $2$.
A similar cutting is constructed at any dual step.

If $\kappa$ is a full-dimensional primal cell of the cutting, constructed
during some stage of the recursion, then each surface $\sigma_{a,a'}$
(which, as we recall, is the union of its strata) either intersects~$\kappa$ 
or is disjoint from $\kappa$. If $\kappa$ is of lower dimension, then it may also 
be fully contained in, or partially overlap, some strata of the surfaces $\sigma_{a,a'}$.
As in Section~\ref{sec:1poly}, we unify these possibilities by saying that each cell
$\kappa$ is either fully contained in a single stratum of $\sigma_{a,a'}$, including
also full-dimensional strata (which tile up the complement of the surface), or is 
crossed by $\sigma_{a,a'}$, meaning that $\sigma_{a,a'}$ intersects $\kappa$ but no stratum of it contains $\kappa$.
Surfaces that cross $\kappa$ are handled recursively.

Similar and symmetric situations arise in the dual setup.

The recursive point-location mechanism can now proceed following 
standard techniques (see below for details). 
Its output is, similar to Section~\ref{sec:1poly}, a union of complete 
bipartite graphs, each of the form $\Sigma_\kappa\times P_\kappa$, 
where $\kappa$ denotes some primal or dual cell of the cutting constructed at some stage of 
the recursion, where in the primal, $\Sigma_\kappa$ is a set of primal surfaces 
that have one (possibly full-dimensional) stratum that fully contains $\kappa$
(there can be at most one such stratum per surface),
and $P_\kappa$ is the set of primal points contained in $\kappa$, 
and where symmetric definitions apply in the dual context.

Lemma~\ref{inv-sigma-sigmastar} implies, similar to the arguments in
Section~\ref{sec:1poly}, that for each bipartite graph 
$\Sigma_\kappa\times P_\kappa$ the following properties hold.
If the graph was constructed in a primal phase then, 
for each~$\sigma_{a,a'}\in\Sigma_\kappa$, each of the sequences 
$\Xi(a,b)$, $\Xi(a',b')$ has a fixed combinatorial structure,
and each comparison between a root in $\Xi(a,b)$ and a root in $\Xi(a',b')$ has
a fixed outcome, for all $(b,b')\in\kappa$. Hence the merged sequence 
$\Lambda(a,b;\;a',b')$ also has a fixed combinatorial structure.
Indeed, any change in any of these invariants
must occur when (in the primal) $(b,b')$ crosses~$\sigma_{a,a'}$, or, when 
$\kappa$ is not full-dimensional, also when $(b,b')$ varies in some fixed 
stratum of $\sigma_{a,a'}$ that partially overlaps $\kappa$, and reaches the 
boundary of this stratum. In either case, by construction, $\sigma_{a,a'}$ 
does not participate in the graph $\Sigma_\kappa\times P_\kappa$.

Symmetric properties hold when $\kappa$ is constructed in the dual phase. 

Recall again the order of quantifiers: We are not claiming that the outcomes
of the comparisons between the roots are the same
for every surface in $\Sigma_\kappa$ and every point in $P_\kappa$, but
rather that they are the same for each fixed surface, uniformly over the points.
That is, different surfaces may have different outcomes, but these outcomes 
are fixed for each surface, over the entire cell~$\kappa$.

To determine these outcomes, we use the following technique.
Suppose we are in the primal, and consider a cell $\kappa$ of the partition
at some recursive step. As argued, for each $\sigma_{a,a'}\in\Sigma_\kappa$, the outcomes
of all the comparisons between the roots of $\Xi(a,b)$ and those of $\Xi(a',b')$,
as well as the combinatorial structure of each of these two sequences,
are independent of the point $(b,b')\in\kappa$. We therefore pick an
arbitrary representative point $(b_\kappa,b'_\kappa)\in\kappa$, not 
necessarily from $B\times B$, and compute, explicitly and by brute force, 
the required outcomes (for the pairs $(a,b_\kappa)$ and $(a,b'_\kappa)$), 
which therefore give us all the outcomes for all the pairs arising from 
$\{\sigma_{a,a'}\}\times P_\kappa$. By repeating
this for every $\sigma_{a,a'}\in\Sigma_\kappa$, we resolve all the comparisons
encoded in $\Sigma_\kappa\times P_\kappa$. A fully symmetric procedure applies
in the dual context. Repeating this step for all cells $\kappa$, both primal and dual,
we resolve all comparisons needed to sort the $c$-roots. The total time for this step
is proportional to the total size of the vertex sets of the complete bipartite graphs,
and we will bound this quantity below.

\paragraph*{Searching with the points of $C$.}
We next search the structure with every $c\in C$. 
We only want to visit subproblems $(\tau,\tau')$ where there 
might exist $a\in\tau$ and $b\in\tau'$ (not necessarily from $A_\tau$ 
or from $B_{\tau'}$), such that $F(a,b,c) = G(a,b,c) = 0$. To find these cells, 
and to bound their number, we consider the two-dimensional surface 
\[
  \pi_{c} \coloneqq \{(a,b)\in \reals^4 \mid F(a,b,c) = G(a,b,c) = 0 \} ,
\]
and our goal is to enumerate, and bound the number of the bottom-level cells $\tau\times\tau'$ 
in the hierarchical partition of $A\times B$ that are crossed by $\pi_c$.

Recall that we have assumed that $\pi_c$ has good fibers, for each $c \in C$
(and that we have already argued that this is indeed the case for complex collinearity testing).
Thus, by Theorem~\ref{thm:hier_partition_CP}(ii) (with~$r=g$), the number of these 
cells is $O((n/g)^{1+\eps})$, and we can find all of them in time $O((n/g)^{1+\eps})$.

Summing over all the $n$ values of $c$, and denoting by~$n_{\tau,\tau'}$ 
the number of surfaces $\pi_c$ that cross~$\tau\times \tau'$, we have
${\displaystyle \sum_{\tau,\tau'} n_{\tau,\tau'} = O(n^{2+\eps}/g^{1+\eps})}$,
for any $\eps>0$. Thus computing all such surface-cell crossings, over all 
$c \in C$, costs $O(n^{2+\eps}/g^{1+\eps})$ time. The cost of searching 
with any specific $c$ in any specific sequence is $O(\log g)$, unless the relevant
cell $\tau\times\tau'$ is met by the exceptional surface $\zeta$, in which case
we inspect all its $O(g^2)$ elements. Hence, using the preceding arguments, the 
overall cost of searching with the elements of $C$ through the structure is 
$$
O\left( n^{3/2+\eps}g^{3/2-\eps} + \frac{n^{2+\eps}}{g^{1+\eps}} \right) ,
$$
with a slightly larger $\eps$.

\paragraph*{Preprocessing: Sorting all the root sequences.}

In order to construct the sorted sequences $\Lambda_{\tau,\tau'}$ of 
the roots $\rho_i(a,b)$, over all pairs $(a,b) \in \tau\times\tau'$, and
over all pairs $(\tau,\tau')$ of bottom-level cells, we use a batched 
point-location strategy, similar to the one in~\cite{BCILOS-17} and in the previous section.
That is, we perform $O(n^{1+\eps}g^{1-\eps})$ point-location queries in an arrangement
of $O(n^{1+\eps}g^{1-\eps})$ piecewise-algebraic 3-surfaces of constant degree in~$\reals^4$,
using a primal-dual approach in which we also swap the roles of points and surfaces, using
the duality described above.

The output of this algorithm is a compact representation 
for the signs of the differences $x(\rho_i(a,b)) - x(\rho_j(a',b'))$ and $y(\rho_i(a,b)) - y(\rho_j(a',b'))$ (where $a, a'\in A_\tau$, 
$b, b'\in B_{\tau'}$, over all pairs of cells $\tau$, $\tau'$, for\footnote{%
  Note that the information collected so far also determines the values $k_{a,b}$ and $k_{a',b'}$.}
$i=1,\ldots,k_{a,b}$, $j=1,\ldots,k_{a',b'}$), given as an edge-disjoint union of 
complete bipartite graphs of the form $\Sigma_{\kappa}\times P_{\kappa}$,
where $\kappa$ is a primal cell, $\Sigma_{\kappa}$ is a set
of surfaces such that either all their (at most three-dimensional) strata miss $\kappa$ or one of
their strata fully contains $\kappa$ (when $\kappa$ is lower-dimensional), 
and $P_{\kappa} \subseteq P$ is the set of points in $\kappa$. 
Symmetric properties apply to graphs constructed at the dual stages.

We show, in the following lemma, that the overall complexity of this representation, measured by the
total size of the \emph{vertex sets} of these graphs, as well as the time
to construct it, are only $O\left((ng)^{8/5+\eps}\right)$, where the $\eps > 0$ 
here is slightly larger than the prescribed $\eps$. 

\begin{lemma}
  \label{lem:85}
  One can perform batched point location of the points of $P$ within the arrangement
  $\A(\Sigma)$, and obtain the above complete bipartite graph representation of the output,
  in $O\left((ng)^{8/5+\eps}\right)$ randomized expected time and storage (where the storage
  here is measured by the overall size of the vertex sets of the graphs) in the algebraic
  decision-tree model, for any prescribed $\eps > 0$, where the constant of proportionality 
  depends on $\eps$ and on the degrees of $F$ and $G$.
\end{lemma}
\begin{proof}
The problem is symmetric in the roles of $(a,a')$ and $(b,b')$,
and therefore has also a symmetric dual version, already discussed above, in which
the relevant pairs $(a,a')$ become points in $\reals^4$, and the relevant pairs
$(b,b')$ become 3-surfaces in $\reals^4$ (that we have already defined earlier).

Put $m \coloneqq n^{1+\eps} g^{1-\eps}$. Recall that 
$\abs{P}$, $\abs{\Sigma} = O(n^{1+\eps} g^{1-\eps}) = O(m)$.
Choose a sufficiently large constant parameter $r > 0$,   
and construct a \emph{$(1/r)$-cutting} for the surfaces in $\Sigma$ 
(which, as we recall, are strata, of dimension at most $3$, that are contained 
in the original surfaces), which is a decomposition of $\reals^4$ into relatively 
open vertical constant-complexity pseudo-prisms (or \emph{prisms}, for short) of 
dimensions~$0,\ldots,4$, each of which is crossed by at most $m/r$~surfaces 
of~$\Sigma$. (For lower-dimensional prisms~$\Delta$, there may be strata that 
fully contain~$\Delta$, and in general we have no control over their number, 
but these surfaces are only used in forming our output graphs, and are not
processed any further---see below.)
Using standard properties of $(1/r)$-cuttings~\cite{CF-90}, combined with 
the analysis of vertical decomposition in four dimensions, as given 
in~\cite{Koltun-04} (see also~\cite{SA-95}), it follows that such a 
decomposition can be constructed in randomized expected time $O(m \poly(r))$ 
(where $\poly(\cdot)$ denotes a polynomial function), and the overall 
number of prisms of all dimensions is $O(r^{4+\eta})$, for any $\eta > 0$, 
where the constant of proportionality depends on $\eta$ and on the degree of the surfaces in $\Sigma$
(which is determined by the degrees of~$F$ and $G$). Each prism, of any dimension, is crossed
(i.e., intersected by but not contained in) at most $m/r$ strata.
We comment that, in the process of constructing the cutting, we take a
random sample of the original surfaces (and not of their strata). We then take each sampled
surface and break it into its strata, and then form the vertical decomposition of the arrangement
of these strata.

For each prism $\kappa$ of the decomposition, let $P_\kappa \subseteq P$ be 
the subset of points of $P$ contained in $\kappa$. If $\abs{P_{\kappa}} > m/r^4$, we
further partition $\kappa$ into subcells, say by slicing it by
hyperplanes orthogonal to the $x_1$-axis, so that each subcell contains at most $m/r^4$ points.
With a slight abuse of notation, we continue to denote these subcells by $\kappa$
and the corresponding subsets by $P_\kappa$. It is easy to verify that the total 
number of these subcells and subsets, over all original prisms, is still
at most $O(r^{4+\eta})$.

For each (refined) cell $\kappa$ of the decomposition, we pass to the dual, with a set 
$P_\kappa^*$ of at most $m/r^4$ 3-surfaces (that is, we first form the original dual surface,
and then take its strata of dimension at most $3$) and a set $\Sigma_\kappa^*$ of at most $m/r$ points (if we have in the primal one or more strata from the same original surface that cross $\kappa$, the dual point corresponds to this entire original surface).
We apply a similar partitioning to these sets, obtaining $O(r^{4+\eta})$ dual prism cells, each
containing at most $(m/r)/r^4 = m/r^5$ dual points and crossed by at most $(m/r^4)/r = m/r^5$ 
dual surfaces (and prisms can be fully contained in any number of strata). 
Altogether there are $O(r^{8+2\eta})$ subproblems, each involving at most $m/r^5$ 
points and surfaces. We now pass back to the primal, and solve each 
of these subproblems recursively. The recursion terminates at subproblems of size 
(number of points and number of surfaces) smaller than $r$.

At each step of the recursion, whether in the primal or in the dual,
for each cell $\kappa$ and each surface $\sigma_{a,a'}$ that has a 
(possibly full-dimensional) stratum that fully contains $\kappa$,
all the points in $\kappa$ have the same sign with respect to $\sigma_{a,a'}$, 
in the sense that the signs (positive, zero, or negative) of all comparisons 
that involve roots of $\Xi(a,b)$ and of $\Xi(a',b')$, are invariant over all 
points $(b,b')\in\kappa$. We recall however that we do not know a priori 
what are those signs, and different surfaces may have different signs. We use the procedure, 
outlined above, of using one sample point in each cell, to determine all these signs.

This yields, for each cell $\kappa$, say a primal cell,
a complete bipartite graph~$\Sigma_\kappa\times P_\kappa$,
with the above properties. At the bottom of the recursion, we simply
produce a collection of trivial graphs, by a brute-force enumeration,
each consisting of two vertices and one edge.
The collection of all these graphs, produced at
all nodes of the recursion, both primal and dual, constitutes the output of the algorithm.

This leads to simple recurrences, one for the overall size of
the vertex sets of the graphs, and one for the actual cost of 
the procedure. Both recurrences solve to the same asymptotic bound
$O\left(m^{8/5+\eps}\right) = O\left((ng)^{8/5+\eps}\right)$,
for any $\eps > 0$, which is slightly greater (by a small constant factor)
than our prescribed $\eps$, provided that $\eta$ is chosen sufficiently small. 
That is, the total expected time bound to locate the points of $P$ within the 
arrangement of~$\Sigma$ is~$O\left((ng)^{8/5+\eps}\right)$, for any prescribed $\eps > 0$.
\end{proof}

To recap, we note, again, that with the output graphs of Lemma~\ref{lem:85}
available, and the signs that they induce, each of the sets 
$\bigcup\{\Xi(a,b) \mid (a,b)\in A_\tau\times B_{\tau'} \}$, 
over all pairs of cells $\tau$, $\tau'$, can be sorted at no additional cost, 
in the algebraic decision-tree model.

\paragraph*{Putting everything together,}
combining the cost of this preprocessing stage with that of the construction of 
the hierarchical partitions for $A$ and $B$, as well as of searching 
with the elements of $C$ in the sorted order obtained (for free) from the
complete bipartite graph representation, we get total expected running time of
${\displaystyle O\left(n\log{n} + (ng)^{8/5+\eps} + \frac{n^{2+\eps}}{g^{1+\eps}} + n^{3/2+\eps}g^{3/2-\eps} \right)}$.
We now choose $g=n^{2/13}$, and obtain expected running time
of $O\left(n^{24/13+\eps}\right)$, where the implied constant of proportionality depends 
on the degrees of $F$ and $G$ and on $\eps$, and the final $\eps$ is a (small) constant multiple 
of the initially prescribed~$\eps$.

\paragraph*{In summary,}
we now state the main results of this section.  Before doing so, for the convenience of the reader,
we aggregate the assumptions made during the analysis into the following list.

\medskip
\noindent\textbf{Assumptions} (for the general case).

\medskip
\noindent\textbf{(i)} 
The polynomials $F$ and $G$ have \emph{good fibers}, in the sense that, for every $c_0\in C$, 
the surface $\pi_{c_0} \coloneqq \{(a,b)\in \reals^4 \mid F(a,b,c_0) = G(a,b,c_0) = 0 \}$
is two-dimensional and has good fibers, meaning that, for each $c_0\in C$,
the system $F(a,b_0,c_0) = G(a,b_0,c_0) = 0$ has only $O(1)$ solutions in $a$ for each $b_0\in B$, 
possibly excluding $O(1)$ exceptional values of $b_0$, and the system 
$F(a_0,b,c_0) = G(a_0,b,c_0) = 0$ has only $O(1)$ solutions in $b$ for each $a_0\in A$, 
again possibly excluding $O(1)$ exceptional values of $a_0$.
As already discussed, this implies that $F$ and $G$ do not have any nontrivial common factor.

\medskip
\noindent\textbf{(ii)} 
We assume that the singular (three-dimensional) locus $\Omega$, of points $(a,b)$ at which
the system $F(a,b,c) = G(a,b,c) = 0$ has either a double root in $c$, or a root that lies 
at a singular point of one of the curves, or infinitely many roots, or two co-vertical roots, 
has good fibers, in the sense that for each $a\in A$,
with the possible exception of $O(1)$ values, the set 
$\Omega_{(a,\cdot)} \coloneqq \{b\in\reals^2 \mid (a,b)\in\Omega\}$ is 
a one-dimensional curve or a discrete set, and for each $b\in B$, 
with the possible exception of $O(1)$ values, the set 
$\Omega_{(\cdot,b)} \coloneqq \{a\in\reals^2 \mid (a,b)\in\Omega\}$ is 
a one-dimensional curve or a discrete set.

\medskip
\noindent\textbf{(iii)} 
For each $c\in C$, the two-dimensional surface $\pi_c$ meets $\Omega^\infty$ 
in a one-dimensional curve or a discrete set of points.
(This allows us to handle efficiently the problem for $(a,b)\in \Omega^\infty$, as described above.)

\medskip
\noindent\textbf{Remark.}
Assumption~(i) is essential for our analysis. In particular, if $\pi_{c_0}$ is three-dimensional, 
for infinitely many values of $c_0$ (as already discussed, if we have only $O(1)$ such values, 
we could test each of them in turn, by a cheaper procedure), we could face an instance of the single vanishing polynomial 
problem, where $A$, $B$ are arbitrary two-dimensional sets, which would then result in
an instance of the problem in either $2{\times}2{\times}1$ or $2{\times}2{\times}2$ dimensions,
problems that we currently are unable to solve in subquadratic time. (To see such an example,
take $G(a,b,c ) = 1 - c_1^2 - c_2^2$, and place all the points of $C$ on the circle
$c_1^2 + c_2^2 = 1$.)
Assumptions (ii) and (iii) are mainly technical, and serve the purpose of resolving the scenario 
involving pairs $(a,b)\in \Omega$ (and, in particular pairs $(a,b)\in \Omega^\infty$).
We do not know how to handle in general such points, in subquadratic time, when these assumptions do not hold.

We then obtain the following main results of this section. 
\begin{theorem}
  \label{thm:3POL}
  Let $A$, $B$, $C$ be three $n$-point sets in the plane, and let $F$, $G$ be 
  a pair of real constant-degree 6-variate polynomials that satisfy assumptions (i)--(iii)
  made above. Then one can test, in the algebraic decision-tree model, whether there exists 
  a triple $(a,b,c)\in A\times B\times C$, such that $F(a,b,c) = G(a,b,c) = 0$,
  using only $O\left(n^{24/13 + \eps}\right)$ polynomial sign tests (in expectation), for any $\eps > 0$.
\end{theorem}


\begin{corollary}
  \label{cor:complex}
  Let $A$, $B$, $C$ be three sets, each of $n$ complex numbers,
  and let $H$ be a constant-degree bivariate polynomial defined over 
  the complex numbers, so that the real and the imaginary parts of 
  $H$ are a pair of real constant-degree 6-variate polynomials that
  satisfy assumptions (i)--(iii) made above.
  Then one can determine, in the algebraic decision-tree model, whether
  there exists a triple $(a,b,c)\in A\times B\times C$ such that $H(a,b,c) = 0$,
  with only $O\left(n^{24/13 + \eps}\right)$ real-polynomial sign 
  tests, in expectation, for any $\eps > 0$.
\end{corollary}

\medskip
\noindent\textbf{Assumptions} (for complex collinearity testing).

We have made several assumptions for this special case, but most of them do not have to be
required a priori, as they can be easily and efficiently tested for, and in case they are violated
the problem can be solved efficiently. These were the assumptions that 
(i)~$A$, $B$ and $C$ are pairwise disjoint,
(ii)~if any of $\gamma_A$, $\gamma_B$, $\gamma_C$ is a line then it is different from the other two curves,
and (iii)~if any of $\gamma_A$, $\gamma_B$, $\gamma_C$ is a line then it does not contain any point from
the other two sets. 

It is easy (and efficient) to test whether these assumptions hold.
If any of them does not hold, it is trivial to test whether a collinearity exists: 
Indeed, for Assumption (i), any coinciding pair forms a collinear
triple with any point of the third set. 
For (iii), it is easy, and efficient, to collect the constant number of exceptional points,
and to check whether they are involved in any collinearity, and to remove them from their 
sets if no such a collinearity is found.
Finally, for~(ii), assuming (i)~and~(iii), 
if any pair of the curves are coinciding lines, collinearity can occur only if this
common line intersects the third set, which is ruled out by~(iii).

We recall that these assumptions imply the good fibers property for complex collinearity testing.

\begin{corollary}
  \label{cor:colcomp}
  Let $A$, $B$, $C$ be $n$-point sets in the complex $zw$-plane, so that $A$ 
  (resp.,~$B$,~$C$) lies on a curve $\gamma_A$ (resp., $\gamma_B$, $\gamma_C$)
  represented by parametric equations of the form 
  $(z,w) = (f_A(t),g_A(t))$ (resp., $(z,w) = (f_B(t),g_B(t))$, $(z,w) = (f_C(t),g_C(t))$),
  where $f_A$, $g_A$, $f_B$, $g_B$, $f_C$, $g_C$ are constant-degree univariate
  complex polynomials. 
  Then one can determine, in the algebraic decision-tree model, 
  whether there exists a collinear triple $(a,b,c)\in A\times B\times C$,
  with $O\left(n^{24/13 + \eps}\right)$ real polynomial sign tests, in expectation, for any $\eps > 0$.
\end{corollary}

\paragraph*{Discussion and further comments.}

It is now time to see why the approach presented in this section 
fails when each of $A$, $B$, $C$ is two-dimensional,
and we have a single polynomial equation (as in collinearity testing in the real plane). 
We follow here the notation from the present section.
If we only enforce the condition $F(a,b,c) = 0$ (in which case the surface $\pi_{c}$ 
is three-dimensional), the efficiency of the method deteriorates:
The number of cells in the hierarchical partition crossed by $\pi_{c}$ at a single level,
under the most favorable assumptions (in particular, having good fibers) would be
$O(D^3)$, leading to a bound of $O((n/g)^{3/2+\eps})$ on the total number of cells met 
by $\pi_{c}$. Then the cost of the search with the elements of $C$ would now be 
(again, suitably modifying $\eps$)
\[
O\left(n\left(n/g\right)^{3/2+\eps} \right)
= O\left(n^{5/2+\eps} / g^{3/2+\eps}\right) .
\]
Balancing between this cost and the cost of the other point-location step, 
which is close to~$O((ng)^{8/5})$, would require choosing $g \approx n^{9/31}$,
and the overall cost would become roughly~$O\left(n^{64/31}\right)$, which is superquadratic.
This explains why collinearity testing has to be restricted to the case of
$2{\times}1{\times}1$ dimensions (as in Section~\ref{sec:col}).
Even the case of $2{\times}2{\times}1$ dimensions
yields a superquadratic solution in our approach, as can be similarly checked.

\section{Collinearity in Higher Dimensions: The $(d \times (d-1) \times(d-1))$ Case}
\label{sec:higher_dim}

Let $A$, $B$ and $C$ be three sets of $n$ points each, so that $A$ is a 
set of points in $\reals^d$ and each of $B$ and $C$ lies in a \emph{hyperplane}.
The goal is to test, in the algebraic decision-tree model, whether 
$A\times B\times C$ contains a collinear triple. Our approach is to use 
a recursive chain of projections, which ultimately map the points in 
$A$, $B$, and $C$ to some plane, so that each of $B$ and $C$ is mapped to 
a set of points on some respective line, collinearity is preserved, and no new collinearity appears
among the projected points. This is a variant of a projection technique described by 
De~Zeeuw~\cite{deZ-18}.\footnote{%
  We are indebted to Adam Sheffer and Frank de Zeeuw for suggesting this approach.}

We denote by $h_1$, $h_2$ the respective hyperplanes containing $B$ and $C$.
In what follows we assume that (a)~$h_1 \neq h_2$ and (b)~$A \subset \reals^d \setminus (h_1 \cup h_2)$.
The reasons for these assumptions are:

\medskip
\noindent
(a)~We must assume $h_1\neq h_2$, for otherwise collinearities could only involve points 
in $h_1$ and the entire problem would be equivalent to testing $(A\cap h_1)\cup B\cup C$
for collinear triples, which is the unrestricted $((d-1)\times(d-1)\times(d-1))$ version 
of collinearity testing, which we do not know how to solve in subquadratic time, even for $d-1=2$.

\medskip
\noindent
(b)~If we allow any points of $A$ to lie in, say, $h_1$, then testing collinearities 
with such points would be equivalent to testing collinearities among triples in 
$(A\cap h_1)\times B \times (C \cap h_1)$, which is a $((d-1)\times(d-1)\times(d-2))$-dimensional 
variant of the collinearity testing problem (assuming $h_1\ne h_2$),
again an instance that we do not know how to solve in subquadratic time.

Accepting these assumptions, we now show that this setting can 
be reduced to the case of collinearity testing in $2{\times}1{\times}$ dimensions, and we can therefore 
attain the bound in Theorem~\ref{thm:coll}, which does not depend on $d$ 
(except for the constant of proportionality).

We may now also assume that no points of $B \cup C$ lie on the $(d-2)$-flat $h_1\cap h_2$.
Indeed, if $h_1 \cap h_2$ contained, say, points from $B$, then any trichromatic collinearity
involving such a point would have to be contained in $h_2$, but $A\cap h_2 = \emptyset$, by assumption,
so no such triples exist, and we may simply delete all points of $(B\cup C)\cap h_1\cap h_2$.

\begin{lemma}
  \label{lem:projection}
  Let $A$, $B$ and $C$ be three sets of $n$ points each, so that $A$ is a
  set of points in $\reals^d$, for some $d\ge 3$, $B$ lies in a hyperplane $h_1$, 
  and $C$ lies in a different hyperplane $h_2$. Assume that 
  $A \subset \reals^d \setminus (h_1 \cup h_2)$ and that $(B \cup C)\cap h_1 \cap h_2 = \emptyset$.
  Then we can project $A$, $B$, and $C$ to some random hyperplane, so that, with probability $1$,
  (i) this transformation is bijective on $A\cup B\cup C$ and preserves collinearity,
  and (ii) each of the images of $B$ and $C$ is contained in a different $(d-2)$-flat,
  and the image of $A$ lies in the complement of the union of these flats.
\end{lemma}

\begin{proof}
We construct a generic random hyperplane $H$ (say, by picking each of its coefficients 
independently at random from $[0,1]$), and project the points
of $A$, $B$, and $C$ onto~$H$, using the following method.
First suppose that $h_1$, $h_2$ are not parallel, so they
intersect in a $(d-2)$-flat $\pi$. Choose a random point $q$ on $\pi$ 
(which, with probability $1$, will not be contained in $H$).
Project each point $p \in A \cup B \cup C$ onto $H$ by mapping $p$ to the 
intersection point of $H$ and the line $pq$; each of these intersection 
points is unique with probability $1$. Indeed, since $H$ is a random hyperplane, 
it is almost surely not parallel to either $h_1$ or $h_2$, and therefore $pq$ 
must meet $H$ at a unique point, as claimed, implying that this 
projection is well-defined. Let $A^{*}$, $B^{*}$, and $C^{*}$ be the 
images of $A$, $B$, and~$C$ on~$H$, respectively. 

Since $q$ is a random point in~$\pi$, each point  
$p \in A \cup B \cup C$ is mapped almost surely to a distinct point 
on $H$, and therefore this projection is a bijection on $A\cup B\cup C$.
It is easy to verify that this projection almost surely preserves 
collinearity, that is, a triple in the original setting is collinear 
if and only if it is mapped to a collinear triple in $H$. Indeed, the
``only~if'' part is obvious. For the ``if'' part, assume to the contrary
that $(a,b,c)\in A\times B\times C$ is a non-collinear triple that
is mapped to a collinear triple $(a^*,b^*,c^*)$. But then $q$, $a^*$,
$b^*$ and~$c^*$ are all coplanar, lying in a unique common plane, which means that
$q$, $a$, $b$ and $c$, all lying in that plane, are also coplanar, which can happen with
probability $0$. This establishes property~(i).

We next prove property (ii), that is, with probability~1, all the points 
of~$B^{*}$ (resp., of~$C^{*}$) lie on a $(d-2)$-flat in $H$, and these flats 
are distinct, and no point of $A$ is projected to any of these flats.
Consider the case of $B^*$. Since $\pi$ is contained 
in $h_1$, the image of $B$ is contained in $h_1\cap H$, which
is almost surely a $(d-2)$-flat. The same argument holds for $C^*$.
Since we assumed $A \subset \reals^d \setminus (h_1 \cup h_2)$, it is easy to
verify that almost surely the image of $A$ is contained in the complement of
the union of the flats $h_1\cap H$, $h_2\cap H$.

The case where $h_1$ and $h_2$ are parallel can be handled in much 
the same way, taking $\pi$ to be the (projective) $(d-2)$-flat at infinity parallel to 
$h_1$ and $h_2$. Effectively, this means that we choose a random direction
parallel to both $h_1$ and $h_2$, and project the
points of $B$ (resp., $C$) in this direction, within $h_1$ (resp.,
$h_2$) onto the flat $H\cap h_1$ (resp., $H\cap h_2$). It is easily
checked that all the properties hold with probability~1 in this case as well.
\end{proof}

Lemma~\ref{lem:projection} suggests a recursive randomized procedure to 
reduce the $d{\times}(d-1){\times}(d-1)$~case to the $2{\times}1{\times}1$~case. 
Given a $d{\times}(d-1){\times}(d-1)$-dimensional instance, we apply 
Lemma~\ref{lem:projection} $d-2$ times, reducing the dimension
by one at each step, until we reach the planar setup and invoke
Theorem~\ref{thm:coll}, and thus obtain the following result:

\begin{theorem}
  \label{cor:proj}
  Let $A$, $B$ and $C$ be three sets of $n$ points each, where $A$ is a set of points in~$\reals^d$
  and each of $B$, $C$ lies in a distinct respective hyperplane $h_1$, $h_2$.
  Assume that $h_1 \neq h_2$, that $A \subset \reals^d \setminus (h_1 \cup h_2)$,
  and that $(B \cup C)\cap h_1 \cap h_2 = \emptyset$. 
  Then one can test whether $A\times B\times C$ contains a collinear triple, in
  the algebraic decision-tree model, by a randomized algorithm that succeeds with probability $1$,
  and uses only $O\left(n^{28/15+\eps}\right)$ polynomial sign tests, for any~$\eps>0$,
  where the constant of proportionality depends on $\eps$ and on $d$.\footnote{%
    The dependence on $d$ appears only in the part of the recursive projections, whose overall running time is linear in $n$.}
\end{theorem}

As mentioned in the introduction, we have sketched, in the earlier version
of this paper~\cite{AES:prev}, initial results for more general extensions 
of both the single-polynomial and the polynomial-pair vanishing problems 
to higher dimensions, where in the former setup, each of~$B$ and~$C$ is 
contained in an algebraic surface of codimension $1$ and constant degree. 
Unlike the bound in Theorem~\ref{cor:proj}, the bounds that this technique seems to yield 
(and that are stated in the introduction) deteriorate with $d$, but 
remain subquadratic for every $d$. We leave the completion of the work
on these higher-dimensional extensions for future research.

\section*{Acknowledgments}
The authors wish to thank Adam Sheffer and Frank de Zeeuw for suggesting the transformation
used for collinearity testing in Theorem~\ref{cor:proj}. This has considerably simplified and shortened our
original analysis for the ``flat'' $d \times (d-1) \times (d-1)$ case.
We also thank Jean Cardinal, John Iacono, Stefan Langerman, and Aur\'elien Ooms 
for useful discussions on the relation of our work with that in \cite{BCILOS-17}.
A preliminary version of this paper \cite{AES:prev} has appeared in 
{\it Proc.~36th Annual Symposium on Computational Geometry}, 2020.

\appendix

\section{Hierarchical Polynomial Partitioning}
\label{app:hierarchical}

The idea of a hierarchical polynomial partition is very similar to the
earlier constructions of hierarchical cuttings, proposed by Chazelle~\cite{Chazelle-91} and
by Matou\v{s}ek~\cite{Mat} already in the 1990s. Still, it has hardly been used
in the context of polynomial partitionings, mostly because almost all the applications
of this technique to date have been combinatorial, so the issue of efficient algorithmic
construction of the partitioning polynomial seldom arises. One notable early exception
is the work of Agarwal \etal~\cite{AMS-13}, mentioned above, which uses a data structure similar 
to the one developed here, albeit in a different context.
More recent algorithms are by Aronov \etal~\cite{AEZ-19} and by Agarwal \etal~\cite{AAEZ-19}, 
which construct polynomial partitions for a set of varieties.

We begin by presenting the hierarchical approach for planar point sets. Let $P$ be 
a set of $n$ points in the plane, and let $1\le r \le n$ be an integer parameter.
Our goal is to efficiently obtain a hierarchical polynomial partitioning of $P$ into 
$O(n/r)$ subsets, each of size at most~$r$, that has properties similar to the single 
partitioning of Guth and Katz~\cite{GK}---see below.

The following easy (but useful) property is related to, but is much simpler than
Proposition~\ref{prop:GK}.

\begin{proposition}[Partitioning a real algebraic curve; Solymosi and De~Zeeuw~\cite{SdZ-18}]
  \label{prop:curve_partition}
  Let $\gamma \subset \reals^2$ be an algebraic curve of degree $\delta$,
  containing a finite set $Q$. Then there is a subset $X \subset \gamma \setminus Q$
  of $O(\delta^2)$ points, such that $\gamma \setminus X$ consists of $O(\delta^2)$
  arcs, each containing at most $\abs{Q}/\delta^2$ points of $Q$.
  Moreover, each point $p \in X$ has an open neighborhood on $\gamma$ (disjoint from $Q$) such that any
  point of that neighborhood could replace~$p$ without affecting the partitioning property.
  The partition can be constructed in $O(\abs{Q})$ time, where the constant of proportionality
  depends on $\delta$.
\end{proposition}

Proposition~\ref{prop:GK} yields a partitioning with a 
single polynomial $f$ of degree $O(\sqrt{n/r})$. However, when $r$ is relatively small 
(so the degree of $f$ is large), the computation of the partition might be costly.
Concretely, the best algorithm for constructing (an approximation of)
such a polynomial is the algorithm by Agarwal \etal~\cite{AMS-13},
which runs in randomized expected time $O(n^2/r + n^3/r^3)$.

To circumvent this issue, we resort to a hierarchical approach, 
where we recursively construct polynomial partitioning of appropriately chosen
constant degree. That is, let $D \ge 1$ be a sufficiently large constant.\footnote{%
  As we will shortly describe, $D$ depends on the prescribed parameter $\eps$.}
We apply Proposition~\ref{prop:GK} with degree~$O(D)$ and Proposition~\ref{prop:curve_partition} with $\delta = D$,
and obtain a partition of~$P$ into at most $cD^2$ one- and two-dimensional cells, 
for a suitable absolute constant $c>0$, so that each cell contains at most $n/D^2$ points of~$P$.
We then recurse with each two-dimensional (resp., one-dimensional) cell $\tau$ by applying 
Propositions~\ref{prop:GK} and~\ref{prop:curve_partition} 
(resp., only Proposition~\ref{prop:curve_partition})
to $P \cap \tau$. The procedure for the one-dimensional case is trivial to perform.

However, the second-level cells $\tau'$ produced at $\tau$ do not have 
to be contained in $\tau$, so the decomposition of the plane that all the second-level cells 
produce, over all $\tau$, is not necessarily a partition of the plane. 
We handle this as follows: The subset associated with a child cell~$\tau'$ of the partition 
is formed by intersecting $\tau'$ with the set associated with its parent cell $\tau$;
only the points of $P$ in that intersection are passed to the subproblem at $\tau'$.
If this set is empty, we do not create a recursive subproblem at $\tau'$. With this modification, 
the subsets of $P$ associated with the current level of recursion form a partition of $P$.

We obtain, in the second recursive stage, a collection of at most~$c^2D^4$ cells, 
each of which is associated with a subset of at most~$n/D^4$ points of~$P$ which
the cell contains (in addition to other points that the cell might also contain
but which are not associated with it). We continue in this manner recursively, so that at the 
$j$th level of recursion we get at most $c^jD^{2j}$ cells, each associated with a subset
of at most $n/D^{2j}$ points of $P$ which it contains. 
(The situation shares common features with the simplicial partitioning scheme
of Matou\v{s}ek~\cite{Ma:ept}, which is based on standard cuttings.)
The recursion terminates when the number of points in a cell is at most $r$.
This leads to a simple recurrence on the number of cells, which solves to 
the bound $O((n/r)^{1+\eps})$, for any $\eps>0$ ($\eps$ depends on $D$, 
or rather $D$ depends on $\eps$; to achieve a bound with a smaller $\eps$ we need to increase $D$,
and pay for the improved bound with a larger constant of proportionality).
In particular, this is an upper bound on the number of cells obtained at the last level. 
We refer to them as the \emph{bottom-level cells}.

So far, this hierarchical construction makes little sense, as there are many other
(more trivial) ways to partition $P$ in this manner. It is the next observation that
makes the hierarchical partition useful, as it shows that the partitioning almost possesses 
(with an $\eps$-loss in the exponents)
the same properties as does a standard polynomial partition. Specifically, let
$\gamma$ be an algebraic curve of constant degree~$b$, and suppose that $D$ is chosen 
to be sufficiently larger than $b$ (see below). 
We bound the number of cells in the decomposition of $P$ that $\gamma$ 
\emph{reaches}, where we say, as above, that $\gamma$ reaches a cell~$\tau$ if $\gamma$~intersects 
$\tau$ and all its ancestral cells. Bounding the number of these cells proceeds as follows.

At the first level of recursion, $\gamma$ meets at most $bD+1$ cells,
which easily follows from B\'ezout's theorem. Clearly, $\gamma$ reaches all these cells.
Let $X_\gamma(n)$ be an upper bound on the maximum number of cells
in a partition of $n$ points that $\gamma$ reaches.\footnote{%
  This recurrence counts the number of reached cells at all levels of the hierarchy. 
  If we are only interested in the number of bottom-level cells that $\gamma$ reaches, 
  the nonrecursive term $bD+1$ should be removed.}
We then obtain the recurrence
\begin{equation}
  \label{eq:crossing_curve}
  X_\gamma(n) \le (bD+1) X_\gamma(n/D^2) + bD+1 ,
\end{equation}
for $n > r$. For $n\le r$, $X_\gamma(n) = 1$, as no further partitioning is done in this case.
Using induction on $n$, it is easy to verify that the solution is 
$X_\gamma(n) = O((n/r)^{1/2+\eps})$, for any $\eps>0$, as long as we
choose $D$ sufficiently larger than $b^{1/(2\eps)}$.
Thus, as promised, this hierarchical partition has similar properties to polynomial partitioning 
with a single polynomial of degree $O(\sqrt{n/r})$ (up to the extra $\eps$ in the exponent):

\begin{theorem}
  \label{app:thm:hier_partition}
  Let $P$ be a set of $n$ points in the plane, let $1\le r \le n$ be an integer parameter, 
  and $\eps>0$ be an arbitrarily small number. Then the following hold: \\
  (i) There is a hierarchical polynomial partition for $P$ with $O((n/r)^{1+\eps})$ bottom-level cells, 
  each of which is associated with at most $r$ points of $P$ which it contains. 
  The overall number of cells is also $O((n/r)^{1+\eps})$. \\
  (ii) Any algebraic curve $\gamma$ of constant degree reaches (in the meaning 
  defined above) at most $O((n/r)^{1/2+\eps})$ cells at all levels of this partition. \\
  The constants in the bounds depend on $\eps$ and on the degree of $\gamma$.
\end{theorem}

The analysis of algorithms for constructing this partition, as well as its
extensions to higher dimensions, presented next, will be given later in this section.

\paragraph*{Hierarchical partition for Cartesian products of two planar point sets.}
We next present an extension of the decomposition of Solymosi and De~Zeeuw~\cite{SdZ-18}
to hierarchical partitioning.

We are now given two $n$-point sets $P_1$ and $P_2$ in the plane 
and an arbitrary integer parameter $1\le r \le n$. Our goal is to obtain a hierarchical 
polynomial partition for $P_{1,2} \coloneqq P_1 \times P_2$, embedded in $\reals^2\times\reals^2$,
which we regard as $\reals^4$, such that at the bottom level of the recursive partition
we obtain roughly $O(\abs{P_{1,2}}/r^2)$ subsets, each of which is the Cartesian product of a subset
of at most $r$ points of $P_1$ with a subset of at most $r$ points of $P_2$, and thus has size at~most~$r^2$. 
The hierarchical partition that we will construct will be a hierarchy of Cartesian products
of planar partitions (for $P_1$ and for $P_2$); see below for details.
Put $N \coloneqq n^2 = \abs{P_1 \times P_2}$.

We construct a planar hierarchical polynomial partition for each of the sets~$P_1$ and~$P_2$, as described above,
and combine them, level by level, to form the desired hierarchical partitioning of $4$-space.
That is, let $D > 1$ be the (large constant) degree parameter of the partition.
We take the first level of the partition of $P_1$ and the first level of the partition of $P_2$,
and construct their Cartesian product, as described in Section~\ref{sec:prelim},
thereby obtaining at most $c^2D^4$ cells in total (of dimensions~$2$, $3$, and $4$), each containing at most~$N/D^4$ 
points of~$P_{1,2}$ (see Corollary~\ref{cor:partition_CP}). By construction, each set of the partition is the
Cartesian product of a set of at most $n/D^2$ points of $P_1$ and a set of at most $n/D^2$ points of $P_2$.
In the next level of the planar partition of~$P_1$ (resp.,~$P_2$), we have at most $cD^2$ subpartitions,
each resulting from a cell at the first level. 
We now consider all pairs of these subpartitions, one from some cell $\tau_1$ of the top-level
partition of $P_1$ and the other from a cell $\tau_2$ from the top-level partition of $P_2$,
and construct their Cartesian product as above. For each pair of top-level cells under consideration, we get at most 
$c^2D^4$ second-level cells, for a total of at most $c^4D^8$ cells, so that each of them is associated with a subset of
at most $N/D^8$ points of $P_{1,2}$ which it contains (where each of these subsets is a Cartesian product of the kind we want).
In general, at the $j$th recursive step of the partition, 
we consider the $j$th level of the two planar partitions of $P_1$ and $P_2$, and, for each pair of cells,
one from each partition, we construct the Cartesian product of the partitions produced at those cells.
We obtain at most $c^{2j}D^{4j}$ cells, each associated with a subset of at most $N/D^{4j}$ points of 
$P_{1,2}$ which it contains (in addition to other points it may also contain, where each set is the Cartesian product of a set of at most $n/D^{2j}$ points of $P_1$ and a set of at most $n/D^{2j}$ points of $P_2$). 
At the last step, each cell in the partition of $P_1$ (resp.,~$P_2$) 
is associated with a subset of at most $r$ points of $P_1$ (resp., of $P_2$), which it contains,
so their product is associated with a subset of at most $r^2$ points of $P_{1,2}$ which it contains.
The number of bottom-level cells in each of the planar partitions is $O((n/r)^{1+\eps})$, 
and therefore the total number of product bottom-level cells in the four-dimensional
construction is at most $O((N/r^2)^{1 + \eps})$, for any $\eps>0$ (for this bound we need to choose $D$ as a
suitable function of $\eps$), and each of them is associated with a subset of at most $r^2$ points of $P_{1,2}$,
which is the Cartesian product of a subset of at most $r$ points of $P_1$ and a subset of at most $r$ points of $P_2$,
which the cell contains.

We now consider the interaction between the resulting partition and a two-dimensional algebraic surface.
Let $S \subset \reals^4$ be a two-dimensional surface of degree at most $b$, which has good fibers.
As above, we assume that $D$ is chosen to be sufficiently large with respect to $b$.
We bound the number of cells in the hierarchical decomposition of $P_1 \times P_2$ that $S$ \emph{reaches},
in the same meaning as above.
With a slight abuse of notation, let $X_S(N)$ denote the bound on the maximum number of 
product cells in a partition of the Cartesian product, of size $N=n^2$, of two planar point sets 
of $n$ points each, that are reached by $S$. Using an enhanced version of Proposition~\ref{prop:intersect}, 
we then obtain the recurrence\footnote{%
  Again, the nonrecursive term $O(b^2D^2)$ should be dropped if we only care about the number of bottom-level cells that $S$ reaches.}:
\begin{equation}
  \label{eq:surface_crossing}
  X_S(N) = O(b^2 D^2) X_S(N/D^4) + O(b^2 D^2) ,
\end{equation}
for $N > r^2$, and $X_S(N) = 1$ for $N\le r^2$ (again, no partition is done in the latter case). 
Using induction on $N$, it is easy to verify that the solution is 
$X_S(N) = O(({N/r^2})^{1/2+\eps}) = O((n/r)^{1+2\eps})$,
where, as above, we need to choose $D$ larger than $b^{1/(2\eps)}$.
We thus conclude:
\begin{theorem}
  \label{app:thm:hier_partition_CP}
  Let $P_1$, $P_2$ be two $n$-point sets in the plane and put $P_{1,2} = P_1 \times P_2$.
  Let $1\le r \le n$ be an integer parameter.
  Then, for any $\eps>0$, the following hold: \\
  (i) There is a hierarchical polynomial partition of $\reals^4$ for $P_{1,2}$ with 
  $O((n/r)^{2+\eps})$ bottom-level cells, 
  each of which is associated with a subset of at most $r^2$ points of $P_{1,2}$,
  which is the Cartesian product of a set of at most $r$ points of $P_1$ and a set 
  of at most $r$ points of $P_2$, which the cell contains. 
  The number of cells at all levels is also $O((n/r)^{2+\eps})$. The constants of proportionality depend on $\eps$.\\
  (ii) A two-dimensional algebraic surface $S$ of degree at most $b$ with good fibers, where $b$ is a constant, reaches at most
  $O((n/r)^{1+2\eps})$ cells at all levels of the hierarchical partition of $P_{1,2}$.
  The constant of proportionality depends on $\eps$ and on $b$.
\end{theorem}

We next comment that when we have a one-dimensional variety (i.e., an algebraic curve) $\gamma$ of degree at most $b$,
one can show, using similar arguments as above and the fact that at every level of the partition $\gamma$ meets at most $bD$ product
cells, that the maximum number of product cells at all levels (as above)
that are reached by $\gamma$ is $O((n/r)^{1/2+\eps})$.
We thus conclude:

\begin{corollary}
  \label{app:cor:hier_partition_1dim}
  Given the setting of Theorem~\ref{app:thm:hier_partition_CP} and an algebraic curve $\gamma$ of degree at most $b$ (where $b$ is a constant), the number of cells at all levels of the hierarchical partition of $P_{1,2}$ reached by $\gamma$ is $O((n/r)^{1/2+\eps})$, for any $\eps > 0$. The constant of proportionality depends on $\eps$ and on $b$.
\end{corollary}

\paragraph*{Hierarchical partition for the Cartesian product of a point set in $\reals^2$ and a one-dimensional point set.}

We are now given a set $P$ of $n$ points in the plane, a set $Q$ of $n$ points on an algebraic curve
$\gamma \subset \reals^2$ of constant degree $\delta$, and two arbitrary integer parameters $r, s \ge 1$,
which we constrain to satisfy $r, s \le \sqrt{n}$. 
Our goal is to obtain a hierarchical polynomial partition for $P \times Q$, 
now embedded in the three-dimensional space $\reals^2 \times\gamma$,
such that, at the bottom level of the recursive partition, we obtain roughly $O(n^2/rs)$ subsets,
each of size at most $rs$. Put $N \coloneqq n^2 = \abs{P \times Q}$.

Let $D_1 > 1$ be the (sufficiently large constant) degree parameter of the partition of $P$
and let $D_2$ be the degree parameter of the partition of $Q$,
chosen to be a constant sufficiently larger than $\delta$ (a more specific choice
of these parameters is stated in the analysis below).
We construct a planar hierarchical polynomial partition for $P$, as described above,
and construct a hierarchical partition for $Q$ by recursively partitioning at each step some subset $Q'$ of $Q$
into at most $cD_2^2$ sets, each containing at most $\abs{Q'}/D_2^2$ points,
using Proposition~\ref{prop:curve_partition}.\footnote{%
  An inspection of the analysis of Proposition~\ref{prop:curve_partition} shows 
  that one needs to split the curve at its singular points, which is the reason we require $D_2 > \delta$.} 
We stop as soon as each subset of $P$ contains at most $r$ points, and each subset of $Q$ contains at most $s$ points.

To simplify the presentation, we assume, from now on, that $\gamma$ is connected and
non-self-intersecting. Handling more general curves can be done by partitioning
$\gamma$ into maximal connected and non-self-intersecting pieces, and by applying
the analysis to each piece separately, but we do not address this issue any more
in what follows. Consider the parametric representation $\gamma(t)$ of $\gamma$, 
for $t \in \reals$, now regarding $\gamma$ as (a homeomorphic copy of) $\reals$, 
and write $\reals^2\times\reals$ instead of $\reals^2\times\gamma$.

%

Before proceeding to the description of the hierarchical partition for $P \times Q$, 
we describe the structure of the decomposition of $\reals^2\times \reals$ at 
a single level, when we form the Cartesian product of the partition of $P$ with 
the partition of $Q$.
We note that although this structure is similar in spirit to the one obtained for the Cartesian product of two planar
point sets, the setting of a planar point set and a one-dimensional point set was not addressed in~\cite{SdZ-18},
and we therefore give these details for the sake of completeness.

Let $f$ be the polynomial of degree $D_1$ partitioning $P$, 
given in Proposition~\ref{prop:GK}. Put $\zeta = Z(f)$, and let $X \subset \zeta$
be the set of $O(D_1^2)$ points obtained by applying Proposition~\ref{prop:curve_partition} 
to $\zeta$ and $P\cap \zeta$.  Similarly, let $X' \subset \gamma$ be the set of $O(D_2^2)$ partitioning 
points for $\gamma$, obtained by applying Proposition~\ref{prop:curve_partition} to~$\gamma$ and~$Q$.

The partitioning of $\reals^2\times\reals$ has the following structure.
The two-dimensional ``walls'' $\zeta \times \reals$ and $\reals^2 \times X'$
partition $\reals^3$ into $O(D_1^2 D_2^2)$ three-dimensional cells, each being the Cartesian product of
a two-dimensional cell of the partition of $\reals^2$ with a one-dimensional cell of the partition of $\reals$. 
The two-dimensional cells are formed by taking $\zeta \times \reals$ and partitioning it by the
``level curves''~$\zeta \times X'$ and the so-called \emph{$1$-gaps}~$X \times \reals$,\footnote{%
  Following the definition in~\cite{SdZ-18}, a \emph{$k$-gap} is a $k$-dimensional variety that 
  is used to cut out $(k+1)$-dimensional cells into subcells,
  but as it does not contain any points of interest, it does not have to be partitioned further.}
and, similarly, by taking $\reals^2 \times X'$ and partitioning it by the curves $\zeta \times X'$.
Overall, we obtain $O(D_1^2 D_2^2)$ two-dimensional cells.
The one-dimensional cells are formed by partitioning $\zeta \times X'$ by $X \times \reals$,
and the vertices ($0$-dimensional cells of the partition) are the endpoints of the one-dimensional cells.
The overall number of $0$- and $1$-dimensional cells is $O(D_1^2 D_2^2)$.  However, since these cells are empty
of points of $P\times Q$, by construction, we disregard them hereafter.
So we get a total of at most $cD_1^2 D_2^2$ cells, for a suitable constant $c$, and each of these cells
contains at most $\abs{P \times Q}/D_1^2 D_2^2$ points of $P \times Q$,
which form a Cartesian product of a subset of at most $\abs{P}/D_1^2$ points of $P$ and
a subset of at most $\abs{Q}/D_2^2$ points of $Q$,
as the cell is the Cartesian product of a one- or a two-dimensional cell from the partition of $P$
with a one-dimensional cell from the partition of $Q$.


We now iterate over the levels of the hierarchy, and construct the Cartesian product of the cells
in the partition of $P$ with the cells in the partition of $Q$ at each level of the hierarchy. 
This is somewhat similar to the
earlier construction for two planar point sets, with the following main difference: 
Recall that we started with two arbitrary parameters $D_1$, $D_2$. 
However, for convenience of the following analysis
we require that both hierarchical
partitions of $P$ and $Q$ have the same number of levels,
which imposes a constraint on the choice of $D_1$ and $D_2$.  
Specifically, let $k$ be the common number of levels in the hierarchical partitions of $P$, 
and $Q$.
We thus have $D_1^{2k} = n/r$ and $D_2^{2k} = n/s$.
We thus need to choose
${\displaystyle D_2 = D_1^{\frac{\log{(n/s)}}{\log{(n/r)}}}}$.
Due to our assumption that $1 \le r, s \le \sqrt{n}$, it follows that
$1/2\le \frac{\log{(n/s)}}{\log{(n/r)}} \le 2$, and so
$\sqrt{D_1} \le D_2 \le D_1^2$.

The remaining details of the decomposition are similar to the earlier construction for two planar point sets.
That is, at the $j$th recursive step,
we construct the Cartesian product of the partitions, for each pair of (parent) cells, and obtain 
at most $c^{2j}D_1^{2j} D_2^{2j}$ cells (of dimensions $1$, $2$, and $3$), each associated with a subset of
at most $N/D_1^{2j} D_2^{2j}$ points of $P \times Q$ which it contains. As in the previous construction,
we obtain a recurrence that shows that the total number of cells is at most $O((n^{2+\eps}/(rs)^{1 + \eps})$,
for any $\eps>0$ (choosing $D_1$ and $D_2$ to be sufficiently large in terms of $\eps$).

Let $S \subset \reals^2\times\reals$ be a two-dimensional algebraic surface of degree at most $b$ with \emph{good fibers}. 
Recalling its definition from Section~\ref{sec:prelim}, this means that, 
for every point $p \in \reals^2$, excluding $O(1)$ exceptional points,
the fiber $(\{p\} \times \reals) \cap S$ is finite, and for every point $q \in \reals$,
again excluding $O(1)$ exceptional points,
the fiber $(\reals^2 \times \{q\}) \cap S$ is a one-dimensional variety (i.e., an algebraic curve). 
As above, we assume that $D_1$ and $D_2$ are chosen 
to be sufficiently large with respect to $b$. We bound the number of cells in the hierarchical 
decomposition of $P \times Q$ that $S$ \emph{reaches}, a notion defined above.


We first give an upper bound on the number of cells that $S$ intersects at any single level in the hierarchy.
Assume first that there are no exceptional points $q$ of the second kind (that is, where the fiber is $\reals^2 \times \{q\}$).
Using the above notation for $\zeta$, $X$, and $X'$, we intersect $S$ with the two-dimensional walls 
$\zeta \times \reals$, $\reals^2 \times X'$, and with the $1$-gap $X \times \reals$.
Due to the good-fibers property and by applying B\'ezout's Theorem, this yields a single curve 
$\alpha$ on $S$ (contained in~$\zeta\times \reals$) of degree $O(bD_1)$ (to which we 
append $O(1)$ copies $\{p\}\times \reals$ of $\reals$, for the $O(1)$ exceptional 
points~$p$ where we do not have a good fiber),
a collection $\Xi$ of $O(D_2^2)$ pairwise disjoint curves, each of degree~$O(b)$, 
and a discrete set $M$ of $O(D_1^2)$ points located on the curve $\alpha$.
We construct the two-dimensional map on $S$ formed by the overlay
of $\alpha$, the curves in $\Xi$, and the points in $M$. The faces of this map correspond to cells in the
Cartesian product of the partitions of $P$ and $Q$ that $S$ crosses. 
(Several cells of the map can correspond to the same cell of the product.)
We observe that, by B\'ezout's Theorem,
$\alpha$ has at most $O(b^2 D_1^2)$ critical points. Moreover, each curve in $\Xi$ meets $\alpha$ in at most
$O(b^2 D_1)$ points (applying once again B\'ezout's Theorem), for a total of $O(b^2 D_1D_2^2)$ intersections over all
curves in $\Xi$. We thus conclude that $O(b^2 D_1 D_2^2 + b^2 D_1^2)$ bounds the total complexity of the overlay,
and thus $S$ crosses 
\[
O(b^2 D_1 D_2^2 + b^2 D_1^2) = O(b^2 D_1 D_2^2) 
\] 
cells in total (the bound on the right hand side follows since $D_1 \le D_2^2$).


Next, we consider the $O(1)$ possible exceptional points $q$ of the second kind.
This, in particular, implies that $S$ may contain up to $O(1)$ planes of the form $\reals^2 \times \{q\}$.
Using the above considerations, it is easy to verify that each of these planes crosses at most $O(D_1^2)$ cells,
which is therefore also the total asymptotic bound over all exceptional points $q$. 
As above, this bound is subsumed by $O(b^2 D_1 D_2^2)$.

With a slight abuse of notation, let $X_S(n_1, n_2)$ denote a bound on the maximum number of cells in a partition
of the Cartesian product of a planar point set of size $n_1$
and a one-dimensional point set of size $n_2$ (with a total of $n_1n_2$ points),
that are intersected by~$S$. We then obtain the recurrence:\footnote{%
  Here too the nonrecursive term can be dropped if we only care about bottom-level cells.}
\begin{equation}
  \label{eq:surface_crossing_3D}
  X_S(n_1, n_2) = O(b^2 D_1 D_2^2) X_S(n_1/D_1^2,n_2/D_2^2 ) + O(b^2 D_1 D_2^2) ,
\end{equation}
for $n_1 > r$ or, equivalently, $n_2 > s$.
We have $X_S(n_1, n_2) = 1$ otherwise. 
Using induction on $N = n_1n_2$, it is easy to verify that (assuming $b$ to be a constant)
the solution is
\[
X_S(n_1, n_2) = O\left(\frac{(\sqrt{n_1} n_2)^{1+\eps}}{r^{1/2+\eps}s^{1+\eps}} \right) ,
\]
for any $\eps>0$.
Since initially we have $n_1 = n_2 = n$, the bound on the number of cells intersected by~$S$ is thus
\[
O\left(\frac{n^{3/2+\eps}}{r^{1/2+\eps}s^{1+\eps}} \right) .
\]
We thus conclude:

\begin{theorem}
  \label{app:thm:hier_partition_12}
  Let $P$ be a set of $n$ points in the plane, and let $Q$ be a set of $n$ points lying on an algebraic curve
  $\gamma \subset \reals^2$ of constant degree $\delta$.
  Let $1 \le r, s \le \sqrt{n}$ be integer parameters.  Let $\eps>0$.
  Then the following hold: \\
  (i) For any $\eps>0$, there is a hierarchical polynomial partition for $P \times Q$ into 
  $O(n^{2+\eps}/(rs)^{1+\eps})$ bottom-level cells, 
  each of which contains at most $rs$ points of $P \times Q$, which form the Cartesian product
  of a subset of at most $r$ points of $P$ with a subset of at most $s$ points of $Q$. \\
  (ii) Any two-dimensional surface $S$ of degree at most $b$ with good fibers reaches at most
  $O\left(\frac{n^{3/2+\eps}}{r^{1/2+\eps}s^{1+\eps}} \right)$
  cells on all levels of this partition.
\end{theorem}

\paragraph*{Computation time.}

We next show that we can efficiently compute the hierarchical polynomial partition, 
as well as the set of partition cells reached by any given two-dimensional algebraic variety,
in both the four- and the three-dimensional cases, i.e., both the product of two planar sets
and the product of a planar point set and a set of points on a curve.


For the analysis below, we recall our assumption about the computational model from Section~\ref{sec:prelim},
in which the computation of the roots of a real univariate polynomial of degree $b$, and subsequent
comparisons and algebraic manipulations of these roots, can be computed in time that depends only on $b$.

We first consider the time to produce the planar hierarchical partition.
At every level of the partition, we apply Propositions~\ref{prop:GK} and~\ref{prop:curve_partition} in order to produce
a partitioning polynomial $f$, and a partition for $Z(f)$. 
Using the algorithm in~\cite{AMS-13}, a suitable variant of 
$f$ can be computed in expected time $O(nD^2 + D^6) = O(n \poly(D))$.
Concerning the implementation of Proposition~\ref{prop:curve_partition}, a closer inspection of the analysis
of~\cite{SdZ-18} shows that a main ingredient in the computation of this partition relies on finding the
critical points of $\zeta=Z(f)$ (there are $O(D^2)$ such points). By our assumption, this computation
takes $O(1)$ time (by equating~$f$ and its $y$-derivative to $0$), when $D$ is constant.
Omitting any further details, this partition can be computed in $O(n)$ time, where the constant of proportionality
depends on $D$, and thus on $\eps$. 
The overall expected time to construct the hierarchical partition is therefore $O(n)$
at every level, due to the fact that the subproblems at each level come from 
an actual partition of the points. Summing over all
$O(\log_D{n})$ levels of the recursion, we obtain total expected construction time of $O(n\log{n})$.

The computation of the hierarchical partition of the Cartesian product of two planar point sets $P_1$, $P_2$
is straightforward, given the separate planar hierarchical partitions of $P_1$ and $P_2$ (in fact, we can represent it
implicitly by these two partitions, and need not perform any further operations).
Similarly, the computation of the hierarchical partition of the Cartesian product of a planar point set $P$
and a set $Q$ lying on an algebraic curve $\gamma \subset\reals^2$ is also
straightforward, given the individual partitions for $P$ and $Q$. We thus conclude:

\begin{corollary}
  \label{app:cor:construction_time}
  With the above notation, the hierarchical polynomial partitions for $P_1 \times P_2$ 
  (in the $2\times 2$-dimensional case) and for $P \times Q$ (in the $2\times 1$-dimensional case) 
  can be computed implicitly in randomized expected time $O(n \log{n})$ in the uniform model, 
  where the constants of proportionality depend on $D$.
\end{corollary}

Given an algebraic curve $\gamma$ of degree at most $b=O(1)$ and a hierarchical planar partition,
we can compute the set of cells in the partition intersected by $\gamma$ in a straightforward manner.
At each level, we need to compute the intersection points of $\gamma$ with the zero set $\zeta$ 
of the partitioning polynomial.
Once again, using our assumption, this takes $O(1)$
time (where the constant of proportionality depends on $D$, $\eps$, and $b$). Extracting the actual
cells that $\gamma$ crosses can be obtained from these intersection points using a suitable planar map
representation of the partition.
Using a simple recurrence relation resembling~\eqref{eq:crossing_curve}, we can conclude that the overall
computation time, for finding the cells that $\gamma$ crosses, is $O((n/r)^{1/2+\eps})$, for any $\eps>0$.

In order to compute the set of cells in the hierarchical partition of $P_{1,2}$ reached by a two-dimensional
surface $S \subset \reals^4$ of degree at most $b = O(1)$ (with good fibers), we need to apply several
elementary algebraic operations at each level of the partition.
Taking a closer inspection of the analysis in~\cite{SdZ-18}, we put $\zeta_1 = Z(\varphi_1)$, 
$\zeta_2 = Z(\varphi_2)$, and let $X_1 \subset \zeta_1$ (resp, $X_2 \subset \zeta_2$) be the set of
$O(D^2)$ points obtained by applying Proposition~\ref{prop:curve_partition} to $\zeta_1$ and $P_1\cap \zeta_1$
(resp., $\zeta_2$ and $P_2\cap \zeta_2$). Given a surface $S$ as above, at each level in the hierarchical 
partition we need to intersect $S$ with the $3$-dimensional walls $\zeta_1 \times \reals^2$ and 
$\reals^2 \times \zeta_2$, and with the so-called $2$-gaps $X_1 \times \reals^2$ and $\reals^2 \times X_2$.
Due to the good-fiber property, the $3$-dimensional walls intersect $S$ in two respective algebraic curves,
$\gamma_1$, $\gamma_2$, and the $2$-gaps 
intersect $S$ in two respective discrete sets of points $Q_1$ and $Q_2$,
located on $\gamma_1$ and $\gamma_2$, respectively. Next, we need to construct,
for each level of the hierarchical partition, the two-dimensional map on $S$
formed by the overlay of $\gamma_1$, $\gamma_2$ and $Q_1$, $Q_2$, 
where the faces of this map correspond to the cells
in the current level of the hierarchical partition that $S$ intersects. 
The proof of Proposition~\ref{prop:intersect} implies that the complexity of this overlay is $O(b^2 D^2)$,
and our assumption guarantees that the time to compute depends only on $D$ 
(and thus on $\eps$) and $b$.\footnote{%
  This is fairly straightforward, e.g., by splitting the curves on $S$ into monotone arcs 
  and then using a sweep-line algorithm.}
This leads to a simple recurrence relation resembling~\eqref{eq:surface_crossing} on the overall running time,
resulting in the bound $O((N/r^2)^{1/2+\eps}) = O((n/r)^{1+2\eps})$.
We thus conclude:

\begin{corollary}
  \label{app:cor:crossing_time}
  Let $P_1$ and $P_2$ be two planar point sets, each consisting of $n$ points, and let
  $S \subset \reals^4$ be a two-dimensional surface of degree at most $b = O(1)$, with good fibers.
  Then the set of cells in the hierarchical polynomial partition of $P_{1,2}$, 
  as constructed in Theorem~\ref{app:thm:hier_partition_CP}, that are reached by $S$
  can be computed in $O((n/r)^{1+2\eps})$ time, for the prescribed parameter $\eps>0$
  of the hierarchical partition, where the constant of proportionality depends on $\eps$ and $b$.
  When $S$ is replaced by an algebraic curve (of degree at most $b = O(1)$), the bound becomes $O((n/r)^{1/2+\eps})$.
\end{corollary}

Regarding hierarchical partitions of the product of a planar point set $P$ with a one-dimensional point set $Q$,
using similar considerations as above, and recalling our discussion leading to the second part of
Theorem~\ref{app:thm:hier_partition_12}, we can show
that the set of child cells, at a single step of the recursive partition, reached by a two-dimensional surface
$S \subset \reals^3$, with good fibers, that has reached the parent cell, 
can be computed in $O(1)$ time (with a constant of proportionality
that depends on $\eps$ and $b$). This leads to a recurrence formula, similar to those obtained
in the proof, which implies that the overall running time to compute the set of cells reached 
by $S$ is $O\left(\frac{n^{3/2+\eps}}{r^{1/2+\eps}s^{1+\eps}} \right)$.
We thus conclude:

\begin{corollary}
  \label{app:cor:crossing_time_12}
  Let $P$ be a planar point set of $n$ points and $Q$ a set of $n$ points
  on a curve of constant degree, and let $S \subset \reals^3$ be a 
  two-dimensional surface of degree at most $b = O(1)$, with good fibers.
  Then the set of cells in the hierarchical polynomial partition of $P \times Q$, 
  as constructed in Theorem~\ref{app:thm:hier_partition_12}, that are reached by $S$
  can be computed in $O\left(\frac{n^{3/2+\eps}}{r^{1/2+\eps}s^{1+\eps}} \right)$
  time, for the prescribed parameter $\eps > 0$ of the hierarchical partition,
  where the constant of proportionality depends on $\eps$ and $b$.
\end{corollary}

\end{document}